\documentclass[11pt]{amsart}
\usepackage{setspace}
\usepackage{amsmath}
\usepackage{amstext}
\usepackage{amssymb}
\usepackage{amsthm} 
\usepackage[foot]{amsaddr}
\usepackage{mathrsfs}
\usepackage{xcolor}
\usepackage{enumerate}   
\usepackage{tikz}
\usepackage{geometry}
\newtheorem{thm}{Theorem}[section]

\newtheorem{lem}[thm]{Lemma}
\newtheorem{prop}[thm]{Proposition}
\theoremstyle{definition}
\newtheorem{defn}[thm]{Definition}
\newtheorem{rem}[thm]{Remark}
\newtheorem{example}[thm]{Example}
\numberwithin{equation}{section}

\newcommand{\be}{\begin{equation}}
\newcommand{\ee}{\end{equation}}
\newcommand{\ba}{\begin{aligned}}
\newcommand{\ea}{\end{aligned}}
\newcommand{\cF}{\mathcal{F}}
\newcommand{\bF}{\mathbb{F}}

\newcommand{\R}{\mathbb{R}}
\newcommand{\N}{\mathbb{N}}
\newcommand{\EE}{\mathbb{E}}
\newcommand{\PP}{P}

\newcommand{\cB}{\mathcal{B}}
\newcommand{\cC}{\mathcal{C}}
\newcommand{\cL}{\mathcal{L}}
\newcommand{\cI}{\mathcal{I}}
\newcommand{\cP}{\mathcal{P}}
\newcommand{\cU}{\mathcal{U}}

\newcommand{\cS}{\mathcal{S}}
\newcommand{\cD}{\mathcal{D}}
\newcommand{\cR}{\mathcal{R}}
\newcommand{\cV}{\mathcal{V}}
\newcommand{\cY}{\mathcal{Y}}

\newcommand{\ind}{\mathbf{1}}
\newcommand{\strat}{\Theta}
\newcommand{\strathat}{\widehat{\strat}}
\newcommand{\proj}{{\rm p}}
\newcommand{\ud}{\mathrm{d}}
\DeclareMathOperator{\essinf}{ess\,inf}
\DeclareMathOperator{\esssup}{ess\,sup}

\title[Arbitrage concepts under trading restrictions in discrete-time]{Arbitrage concepts under trading restrictions in discrete-time financial markets}
\author[C. Fontana]{Claudio Fontana}
\author[W. J. Runggaldier]{Wolfgang J. Runggaldier}
\address{Department of Mathematics ``Tullio Levi - Civita'', University of Padova, Italy.}
\email{fontana@math.unipd.it; runggal@math.unipd.it}
\date{\today}
\keywords{Trading constraints; market viability; arbitrage of the first kind; num\'eraire portfolio.}
\thanks{{\em JEL classification}: C02, C61, G11, G13.\\
Financial support from the University of Padova (research programme BIRD190200/19) and the Europlace Institute of Finance is gratefully acknowledged. Declarations of interest: none.}

\begin{document}

\maketitle

\begin{abstract}
In a discrete-time setting, we study arbitrage concepts in the presence of convex trading constraints.
We show that solvability of portfolio optimization problems is equivalent to absence of arbitrage of the first kind, a condition weaker than classical absence of arbitrage opportunities. 
We center our analysis on this characterization of market viability and derive versions of the fundamental theorems of asset pricing based on portfolio optimization arguments.
By considering specifically a discrete-time setup, we simplify existing results and proofs that rely on semimartingale theory, thus allowing for a clear understanding of the foundational economic concepts involved.
We exemplify these concepts, as well as some unexpected situations, in the context of one-period factor models with arbitrage opportunities under borrowing constraints.
\end{abstract}

\section{Introduction}		\label{sec:intro}

The notions of arbitrage, market viability and state-price deflators are deeply connected and play a foundational role in financial economics and mathematical finance.
Starting from the seminal works \cite{Ross77,Ross78}, the connections between these three concepts represent the essence of the fundamental theorem of asset pricing.\footnote{We refer the reader to \cite{Schachermayer10} for an excellent overview of the main steps in the development of discrete-time and continuous-time versions of the fundamental theorem of asset pricing. The terminology {\em fundamental theorem of asset pricing} has been introduced in \cite{DybvigRoss87}.}
In frictionless discrete-time financial markets, if no trading restrictions are imposed, the appropriate no-arbitrage concept takes the classical form of absence of arbitrage opportunities ({\em no classical arbitrage}). By the fundamental theorem of asset pricing of \cite{HarrisonKreps79,HarrisonPliska81} (extended to general probability spaces in \cite{DMW90}), this is equivalent to the existence of an equivalent martingale measure, whose density acts as a state-price deflator.
Moreover, always in the absence of trading restrictions, the results of \cite{RasonyStettner06} imply that no classical arbitrage is equivalent to market viability, intended as the solvability of portfolio optimization problems.
No classical arbitrage thus represents the minimal economically meaningful no-arbitrage requirement for a frictionless discrete-time financial market.

In the presence of trading restrictions, these results continue to hold true as long as the set of constrained strategies is a cone, provided that equivalent martingale measures are replaced by equivalent supermartingale measures (see \cite[Theorem 9.9]{FollmerSchied} and Theorem \ref{thm:arb} below). 
However, many practically relevant trading restrictions, such as borrowing constraints or the possibility of limited short sales, correspond to {\em convex non-conic} constraints.
In this case, as it will be shown below, market viability is no longer equivalent to no classical arbitrage, but rather to the weaker condition of {\em no arbitrage of the first kind} (NA$_1$). 
Under convex trading restrictions, NA$_1$ represents therefore the minimal economically meaningful concept of no-arbitrage and is equivalent to the existence of a {\em num\'eraire portfolio} or, more generally, a {\em supermartingale deflator}.

The NA$_1$ condition, introduced under this terminology in \cite{Kardaras10}, corresponds to the absence of positive payoffs that can be super-replicated with an arbitrarily small initial capital and is equivalent to the no unbounded profit with bounded risk condition studied in the seminal work \cite{KaratzasKardaras07} (see also \cite{Fontana15} for an analysis of no-arbitrage conditions equivalent to NA$_1$).
In continuous-time, a complete theory based on NA$_1$ has been developed in a general semimartingale setting starting with \cite{KaratzasKardaras07}, also allowing for convex (non-conic) constraints. The connection between NA$_1$ and market viability has been characterized in \cite{ChoulliDengMa15} in an unconstrained semimartingale setting (see also \cite{CCFM17} for further results in this direction).

Scarce attention has, however, been specifically paid to NA$_1$ in discrete-time models, despite their widespread use in economic theory.
This is also due to the fact that, for discrete-time markets with conic constraints, there is no distinction between NA$_1$ and no classical arbitrage (see Remark \ref{rem:compare_arb} below). 
To the best of our knowledge, the only works that specifically address discrete-time models by relying on no-arbitrage requirements weaker than no classical arbitrage are \cite{ElsingerSummer01} and \cite{KS09}. 
In a one-period model on a finite probability space, \cite{ElsingerSummer01} show that limited forms of arbitrage may coexist with market equilibrium under convex constraints (see Remark \ref{rem:ElsingerSummer} below for a more detailed discussion). 
Closer to our setting, \cite{KS09} derive the central results of \cite{KaratzasKardaras07} on the num\'eraire portfolio in a one-period setting.

The present paper intends to fill this gap in the literature, in the framework of general discrete-time models with convex (not necessarily conic) constraints. Compared to \cite{ElsingerSummer01,KS09}, we develop a complete theory of asset pricing based on NA$_1$, also in the case of multi-period models with random convex constraints.
We prove that market viability is equivalent to NA$_1$, thereby showing that no classical arbitrage may pose unnecessary restrictions in the case of non-conic constraints. 
Building our analysis on this central result, we derive versions of the fundamental theorem of asset pricing, study the valuation of contingent claims and discuss non-trivial examples of our theory in the context of general factor models.
We make a systematic effort to provide direct and self-contained proofs based on portfolio optimization arguments.
The simplicity of the discrete-time structure allows for a clear understanding of the economic concepts involved, avoiding the technicalities of the continuous-time semimartingale setup.


The paper is divided into three sections, whose contents and contributions can be outlined as follows.
In Section \ref{sec:1p}, we consider a general one-period setting. Extending the analysis of \cite{KS09}, we prove the equivalence between NA$_1$ and the solvability of portfolio optimization problems (market viability), thus establishing the minimality of NA$_1$ from an economic standpoint. 
This enables us to obtain a direct proof of the characterization of NA$_1$ in terms of the existence of the num\'eraire portfolio or, more generally, a deflator. 
We show that NA$_1$ leads to a dual representation of super-hedging values and a characterization of attainable claims, and permits to rely on several well-known hedging approaches in constrained incomplete markets, even in the presence of arbitrage opportunities.
Besides its pedagogical value, the one-period setting introduces several techniques that will be important for the analysis of the multi-period case.

Section \ref{sec:arbitrage_model} illustrates the theory in the context of factor models with borrowing constraints. 
We introduce a general factor model, where a single factor is responsible of potential arbitrage opportunities. 
In this setting, the NA$_1$ condition and the set of arbitrage opportunities admit explicit descriptions in terms of the factor loadings. When NA$_1$ holds but no classical arbitrage does not, we show the existence of a maximal arbitrage strategy.
These results can be easily visualized in a two-dimensional setting, which enables us to provide examples of situations where, despite the existence of arbitrage opportunities, it is not necessarily optimal to invest in them.
The analysis of this section clarifies the interplay between the support of the asset returns distribution, their dependence structure and the borrowing constraints.

Finally, Section \ref{sec:multiperiod} generalizes the central results of Section \ref{sec:1p} to a multi-period setting with random convex constraints. 
We derive several new characterizations of NA$_1$, showing that it holds globally if and only if it holds in each single trading period, and prove its equivalence to market viability. 
The most general result on the solvability of portfolio optimization problems in discrete-time was obtained in \cite{RasonyStettner06}, relying on no classical arbitrage. Our Theorem \ref{thm:utility_mp} extends this result by introducing trading restrictions and weakening the no-arbitrage requirement to the minimal condition of NA$_1$ (in turn, our proofs of Theorems \ref{thm:utility} and \ref{thm:utility_mp} are inspired from \cite{RasonyStettner06}).
By generalizing the one-period analysis, we then give an easy proof of the equivalence between NA$_1$, the existence of the num\'eraire portfolio and the existence of a supermartingale deflator, for general discrete-time models with random convex constraints.

We close this introduction by briefly reviewing some related literature,  limiting ourselves to selected contributions that are specifically connected with the present discussion.
Relying on the concept of no classical arbitrage, the fundamental theorem of asset pricing with constraints on the amounts invested in the risky assets is proved in \cite{PhamTouzi99} in the case of conic constraints (see also \cite{KoehlPham00,Pham00} for valuation and hedging problems in that setting) and in \cite{Brannath97} in the case of convex constraints. 
The specific case of short-sale constraints is treated in the earlier work \cite{JouiniKallal95}. General forms of conic constraints have been considered in \cite{Napp03}, extending the analysis of \cite{PhamTouzi99}.
In the case of convex constraints on the fractions of wealth invested, as considered in the present work, versions of the fundamental theorem of asset pricing based on the usual notion of no classical arbitrage  are given in \cite{CPT01,EST04,Rokhlin05}. In comparison to the latter contributions, we choose to work with the weaker concept of NA$_1$, due to its equivalence to market viability.
In an unconstrained setting, the connection between no classical arbitrage and market viability is studied in \cite{RasonyStettner05,RasonyStettner06}, generalized in \cite{Nutz16} under model uncertainty.
In the presence of model uncertainty and convex portfolio constraints, \cite{BayraktarZhou17} prove a version of the fundamental theorem of asset pricing based on a robust generalization of the notion of no classical arbitrage.
Finally, we mention the recent work \cite{BCL19}, where super-hedging has been studied under a weak no-arbitrage condition, called absence of immediate profits. However, the latter condition does not suffice to ensure market viability.

\section{The single-period setting}		\label{sec:1p}

We consider a general financial market in a one-period economy, where $d$ risky assets are traded, together with a riskless asset with constant price equal to one. We assume that asset prices are discounted with respect to a baseline security and are represented by the vector $S_t=(S^1_t,\ldots,S^d_t)^{\top}\in\R^d_+$, for $t=0,1$, expressed in terms of returns as
\[
S^i_1 = S^i_0(1+R^i),
\qquad\text{ for all }i=1,\ldots,d,
\]
where $R=(R^1,\ldots,R^d)^{\top}$ is a $d$-dimensional random vector on a given probability space $(\Omega,\cF,\PP)$ such that $R^i\geq-1$ a.s., for all $i=1,\ldots,d$.
We denote by $\cS$ the support of the distribution of $R$, namely the smallest closed set $A\subset\R^d$ such that $\PP(R\in A)=1$ (see \cite[Proposition 1.45]{FollmerSchied}). We also denote by $\cL$ the smallest linear subspace of $\R^d$ containing $\cS$ and by $\cL^{\perp}$ its orthogonal complement in $\R^d$. The orthogonal projection of a vector $x\in\R^d$ on $\cL$ is denoted by $\proj_{\cL}(x)$.

\subsection{Trading restrictions}	\label{sec:constraints}

Trading strategies are denoted by vectors $\pi\in\R^d$. We write $V^{\pi}_t(v)$ for the wealth at time $t$ generated by strategy $\pi$ starting from initial capital $v>0$, with
\[
V^{\pi}_0(v) = v
\qquad\text{ and }\qquad
V^{\pi}_1(v) = v(1+\langle\pi,R\rangle),
\]
where $\langle\cdot,\cdot\rangle$ denotes the scalar product in $\R^d$.
With this notation, a trading strategy $\pi$ represents fractions of wealth held in the $d$ risky assets, with the remaining fraction $1-\langle\pi,\mathbf{1}\rangle$ being held in the riskless asset. 
For $i=1,\ldots,d$, a negative value of $\pi_i$ corresponds to a short position in the $i$-th risky asset, whenever short sales are allowed. Similarly, $\langle\pi,\mathbf{1}\rangle<1$ corresponds to a positive investment in the riskless asset, while $\langle\pi,\mathbf{1}\rangle>1$ corresponds to borrowing from the riskless asset.

Note that $V^{\pi}_t(v)=vV^{\pi}_t(1)$, for all $v>0$ and $t=0,1$.
In the following, we shall use the notation $V^{\pi}_t:=V^{\pi}_t(1)$. A trading strategy $\pi$ is said to be {\em admissible} if $V^{\pi}_1\geq0$ a.s. Denoting by $\strat_{\rm adm}$ the set of all admissible trading strategies, it holds that (see, e.g., \cite[Lemma 4.3]{KS99})
\[
\strat_{\rm adm} = \{\pi\in\R^d : \langle\pi,z\rangle\geq-1\text{ for all }z\in\cS\}.
\]
In the terminology of \cite{KaratzasKardaras07,KS09}, the set $\strat_{\rm adm}$ corresponds to the {\em natural constraints} ensuring non-negative wealth. 
Observe that, with the present parametrization, the notion of admissibility does not depend on the initial capital.

Besides the natural constraints, we assume that market participants face additional trading restrictions, represented by a convex closed set $\strat_{\rm c}\subseteq\R^d$. Realistic examples of trading restrictions include the following situations (see also \cite[Section 4]{CPT01} for additional examples):
\begin{enumerate}[(i)]
\item {\em prohibition of short-selling}: $\strat_{\rm c}=\R^d_+$;
\item {\em prohibition of short-selling and borrowing}: $\strat_{\rm c}=\Delta^d$, where $\Delta^d:=\{\pi\in\R^d_+ : \langle\pi,\mathbf{1}\rangle\leq 1\}$;
\item {\em limits to borrowing}: $\strat_{\rm c}=\{\pi\in\R^d : \langle\pi,\mathbf{1}\rangle\leq c\}$, for some $c\geq1$;
\item {\em limited positions in the risky assets}: $\strat_{\rm c}=\prod_{i=1}^d[-\alpha_i,\beta_i]$, for some $\alpha_i,\beta_i>0$, $i=1,\ldots,d$.
\end{enumerate}
Market participants are restricted to choose strategies  belonging to the set $\strat:=\strat_{\rm adm}\cap\strat_{\rm c}$. We refer to such strategies as {\em allowed strategies}.
Observe that, as illustrated by examples (ii) and (iii) above, trading restrictions on the riskless asset can be enforced by  specifying through the set $\Theta_{\rm c}$ suitable restrictions on the total fraction of wealth held in the $d$ risky assets.

In general, the financial market may contain {\em redundant assets}, meaning that different combinations of assets may generate identical portfolio returns. This happens whenever $\cL^{\perp}$ is strictly bigger than $\{0\}$. Indeed, $\rho\in\cL^{\perp}$ if and only if $\langle\pi,R\rangle=\langle\pi+\rho,R\rangle$ a.s. for every $\pi\in\strat$.
In other words, investing according to a strategy $\rho\in\cL^{\perp}$ does not produce any loss or profit and, therefore, does not alter the outcome of any other allowed strategy $\pi$. For this reason, we shall assume that investors are always allowed to choose trading strategies in the set $\cL^{\perp}$, meaning that $\cL^{\perp}\subset\strat_{{\rm c}}$. In turn, this implies that $\cL^{\perp}\subset\strat$. 

To the convex closed set $\strat$, we associate its {\em recession cone} $\strathat$, defined as the set of all vectors $y\in\R^d$ such that $\pi+\lambda y\in\strat$ for every $\lambda\geq0$ and $\pi\in\strat$ (see \cite[Chapter 8]{Rockafellar}). The set $\strathat$ has a clear financial interpretation: it represents the set of all allowed  strategies that can be arbitrarily scaled and added to any other strategy $\pi\in\strat$ without violating admissibility and trading restrictions.
The cone $\strathat$ is closed and, by \cite[Corollary 8.3.2]{Rockafellar}, it holds that
\[
\strathat = \bigl\{\pi\in\R^d : a^{-1}\pi\in\strat \text{ for all }a>0\bigr\} = \bigcap_{a>0}a\strat.
\]
As a consequence of the fact that $\cL^{\perp}$ is a linear subspace of $\strat$, it holds that $\cL^{\perp}\subseteq\strathat$. In turn, the latter property can be easily seen to imply that $\proj_{\cL}(\strat)\subseteq\strat$, i.e., $\proj_{\cL}(\pi)\in\strat$ for all $\pi\in\strat$.

\subsection{Arbitrage concepts}	\label{sec:arbitrage_concepts}

We proceed to recall two important notions of arbitrage. First, we define the set\footnote{The definition of the set $\cI_{\rm arb}$ is coherent with the classical definition of arbitrage (see \cite[Definition 1.2]{FollmerSchied}). 
Indeed,
$\cI_{\rm arb}\neq\emptyset$ if and only if there exists a portfolio $(\vartheta_0,\vartheta)\in\R\times\R^d$ such that $\vartheta_0+\langle\vartheta,S_0\rangle=0$, $\vartheta_0+\langle\vartheta,S_1\rangle\geq0$ a.s. and $P(\vartheta_0+\langle\vartheta,S_1\rangle>0)>0$, with $\vartheta_0$ and $\vartheta$ denoting respectively the number of shares of the riskless asset and of the $d$ risky assets. 
Assuming without loss of generality that $S^i_0>0$, for all $i=1,\ldots,d$, this equivalence follows in a straightforward way by setting  $\vartheta_i=\pi_i/S^i_0$, for $i=1,\ldots,d$, and $\vartheta_0=-\langle\pi,\mathbf{1}\rangle$.
This shows that absence of arbitrage can be equivalently understood as a property of the returns $R$ or of the price couple $(S_0,S_1)$.}
\[
\cI_{\rm arb} :=\bigl\{\pi\in\R^d : \langle\pi,z\rangle\geq0\text{ for all }z\in\cS\bigr\}\setminus\cL^{\perp}.
\]
Under trading restrictions, the set of {\em arbitrage opportunities} is given by $\cI_{\rm arb}\cap\strat$ and consists of all allowed strategies $\pi$ that generate a non-negative and non-null return (see also \cite[Definition 3.5]{KS09}). We say that {\em no classical arbitrage} holds if $\cI_{\rm arb}\cap\strat=\emptyset$.

We now recall a second and stronger notion of arbitrage (see \cite[Definition 1]{Kardaras10}). To this effect, we define as follows the super-hedging value $v(\xi)$ of a non-negative random variable $\xi$:
\be	\label{eq:superhedge_price}
v(\xi) := \inf\bigl\{v>0 : \exists\;\pi\in\strat \text{ such that }v(1+\langle\pi,R\rangle)\geq\xi\text{ a.s.}\bigr\}.
\ee
In the next definition, we denote by $L^0_+$  the family of non-negative  random variables on $(\Omega,\cF)$.

\begin{defn}	\label{def:NA1}
A random variable $\xi\in L^0_+$ with $\PP(\xi>0)>0$ is an {\em arbitrage of the first kind} if $v(\xi)=0$.  We say that {\em no arbitrage of the first kind} (NA$_1$) holds if $v(\xi)=0$ implies $\xi=0$ a.s.
\end{defn}

An arbitrage of the first kind consists of a non-negative non-null payoff that can be super-replicated starting from an arbitrarily small initial capital. 
Observe that NA$_1$ is weaker than no classical arbitrage, as will be explicitly illustrated by the examples considered in Section \ref{sec:arbitrage_model}. The next proposition provides three equivalent formulations of the NA$_1$ condition.

\begin{prop}	\label{prop:NA1}
The following are equivalent:
\begin{enumerate}[(i)]
\item the {\em NA}$_1$ condition holds;
\item $\cI_{\rm arb}\cap\strathat=\emptyset$;
\item $\strathat=\cL^{\perp}$;
\item the set $\strat\cap\cL$ is bounded (and, hence, compact).
\end{enumerate}
\end{prop}
\begin{proof}
$(i)\Rightarrow(ii)$: by way of contradiction, suppose that NA$_1$ holds and there exists $\pi\in\cI_{\rm arb}\cap\strathat$. Then $\xi:=\langle\pi,R\rangle\in L^0_+$ and $\PP(\xi>0)>0$. For every $v>0$, it holds that  $\pi/v\in\strat$ and $v(1+\langle\pi/v,R\rangle)>\xi$ a.s. This implies that $v(\xi)=0$, yielding a contradiction to NA$_1$.\\
$(ii)\Rightarrow(iii)$: we already know that $\cL^{\perp}\subseteq\strathat$. 
Conversely, since $\strathat\subseteq\bigcap_{a>0}a\strat_{\rm adm}$, every element $\pi\in\strathat$ satisfies $\langle\pi,R\rangle\geq0$ a.s.
Condition {\em (ii)} then implies $\langle\pi,R\rangle=0$ a.s., so that $\pi\in\cL^{\perp}$.\\
$(iii)\Rightarrow(iv)$: the set $\strat\cap\cL$ is non-empty, closed and convex. Hence, by \cite[Theorem 8.4]{Rockafellar}, $\strat\cap\cL$ is bounded if and only if its recession cone $\widehat{\strat\cap\cL}$ consists of the zero vector alone. 
Since $\widehat{\strat\cap\cL}=\strathat\cap\cL$, condition $(iii)$ implies that $\widehat{\strat\cap\cL}=\{0\}$, thus establishing property {\em (iv)}.\\
$(iv)\Rightarrow(i)$: by way of contradiction, let $\xi\in L^0_+$ with $\PP(\xi>0)>0$ and suppose that, for all $n\in\N$, there exists $\pi^n\in\strat$ such that $n^{-1}(1+\langle\pi^n,R\rangle)\geq\xi$ a.s. In this case, it holds that $1+\langle\proj_{\cL}(\pi^n),R\rangle\geq n\xi$ a.s., for all $n\in\N$. Since $\PP(\xi>0)>0$ and $\proj_{\cL}(\pi^n)\in\strat\cap\cL$, for every $n\in\N$, this contradicts the boundedness of the set $\strat\cap\cL$.
\end{proof}

The properties stated in Proposition \ref{prop:NA1} admit natural and direct interpretations, which can be formulated as follows:
\begin{enumerate}
\item[{\em (ii)}] there do not exist arbitrage opportunities that can be arbitrarily scaled;
\item[{\em (iii)}] all allowed strategies that can be arbitrarily scaled reduce to trivial strategies;
\item[{\em (iv)}] all allowed strategies not containing degeneracies are bounded.
\end{enumerate}
As shown in Sections \ref{sec:ftaps} and \ref{sec:hedging} below, the compactness property {\em (iv)} is fundamental, since it allows solving optimal portfolio and hedging problems under NA$_1$, even when no classical arbitrage fails to hold.
The equivalence $(i)\Leftrightarrow(iv)$ is therefore the most important novel insight provided by Proposition \ref{prop:NA1}.
The condition $\cI_{\rm arb}\cap\strathat=\emptyset$ appears in \cite{KaratzasKardaras07,KS09} under the name {\em no unbounded increasing profit} (NUIP), where the unboundedness refers to the fact that the arbitrage profit generated by an element of $\cI_{\rm arb}\cap\strathat$ can be scaled to arbitrarily large values.

\begin{rem}	\label{rem:compare_arb}
Under {\em conic} trading restrictions, no classical arbitrage holds if and only if there are no arbitrages of the first kind. This simply follows from the observation that, if $\strat_{\rm c}$ is a cone, then $\cI_{\rm arb}\cap\strathat=\cI_{\rm arb}\cap\strat$. 
This implies that the two arbitrage concepts differ only in the presence of additional restrictions beyond conic (and, in particular, natural) constraints.
\end{rem}

\begin{rem}[On relative arbitrage]	\label{rem:rel_arb}
The arbitrage concepts introduced so far have been implicitly defined with respect to the riskless asset with constant price equal to one. More generally, in the spirit of \cite[Definition 6.1]{FK09}, a strategy $\pi\in\strat$ is said to be an {\em arbitrage opportunity relative to $\theta\in\strat$} if $P(V^{\pi}_1\geq V^{\theta}_1)=1$ and $P(V^{\pi}_1>V^{\theta}_1)>0$ or, equivalently, if $\pi-\theta\in\cI_{\rm arb}\cap(\strat-\theta)$.
If $\theta\in\strathat_{\rm c}$, then $\cI_{\rm arb}\cap(\strat-\theta)=\emptyset$ implies no classical arbitrage (i.e., $\cI_{\rm arb}\cap\strat=\emptyset$). Conversely, if $-\theta\in\strathat_{\rm c}$, then $\cI_{\rm arb}\cap\strat=\emptyset$ implies $\cI_{\rm arb}\cap(\strat-\theta)=\emptyset$. 
It follows that, for every $\theta\in\strathat_{\rm c}\cap(-\strathat_{\rm c})$, no classical arbitrage coincides with absence of arbitrage opportunities relative to $\theta$.\footnote{The condition $\theta\in\strathat_{\rm c}\cap(-\strathat_{\rm c})$ amounts to saying that arbitrary long and short positions in the portfolio $\theta$ are not precluded by the trading restrictions represented by $\strat_{\rm c}$. This condition is conceptually equivalent to the requirement appearing in the definition of num\'eraire adopted in \cite{KST16} (see Definition 10 therein).}
However, there is no general implication between the two conditions $\cI_{\rm arb}\cap\strat=\emptyset$ and $\cI_{\rm arb}\cap(\strat-\theta)=\emptyset$.
Observe that, unlike arbitrage opportunities, the notion of arbitrage of the first kind is universal, in the sense that it does not depend on a reference strategy $\theta$ (see Definition \ref{def:NA1}). 
The relation between NA$_1$ and relative arbitrage  is further discussed in Remark \ref{rem:NA1_relarb}.
\end{rem}

\subsection{Market viability and fundamental theorems}	\label{sec:ftaps}

The economic relevance of the NA$_1$ condition is explained by its equivalence with the solvability of optimal portfolio problems, as shown in the next theorem. 
We denote by $\cU$ the set of all {\em random utility functions}, consisting of all functions $U:\Omega\times\R_+\rightarrow\R\cup\{-\infty\}$ such that $U(\cdot,x)$ is $\cF$-measurable and bounded from below, for every $x>0$, and $U(\omega,\cdot)$ is continuous, strictly increasing and concave, for a.e. $\omega\in\Omega$.\footnote{For simplicity of notation, we shall omit to denote explicitly the dependence of $U$ on $\omega$ in the following.}
Besides allowing for the possibility of random endowments or state-dependent preferences, the extension to {\em random} utility functions will be needed in the proof of Theorem \ref{thm:numeraire} as well as for the solution of certain hedging and valuation problems (see the last part of Section \ref{sec:hedging}).

\begin{thm}	\label{thm:utility}
The following are equivalent:
\begin{enumerate}[(i)]
\item the {\em NA}$_1$ condition holds;
\item for every $U\in\cU$ such that $\sup_{\pi\in\strat}\EE[U^+(V^{\pi}_1)]<+\infty$, there exists an allowed strategy $\pi^*\in\strat\cap\cL$ such that
\be	\label{eq:opt_utility}
\EE\bigl[U(V^{\pi^*}_1)\bigr] 
= \underset{\pi\in\strat}{\sup}\,\EE\bigl[U(V^{\pi}_1)\bigr].
\ee
\end{enumerate}
\end{thm}
\begin{proof}
$(i)\Rightarrow(ii)$: note first that  \eqref{eq:opt_utility} can be equivalently stated by maximizing over $\strat\cap\cL$, since for every $\pi\in\strat$ it holds that $\langle\pi,R\rangle=\langle\proj_{\cL}(\pi),R\rangle$ a.s. 
By Proposition \ref{prop:NA1}, NA$_1$ implies that $\widehat{\strat\cap\cL}=\strathat\cap\cL=\{0\}$. Hence, in view of \cite[Theorem 27.3]{Rockafellar}, it suffices to show that the proper concave function $u:\strat\cap\cL\rightarrow\R$ defined by
$
u:\strat\cap\cL\ni\pi\mapsto 
u(\pi):=\EE[U(1+\langle\pi,R\rangle)]
$
is upper semi-continuous.
To this effect, we adapt some of the arguments of \cite[Lemma 2.3]{RasonyStettner06} (see also \cite[Lemma 2.8]{Nutz16}). 
Since the set $\strat\cap\cL$ is bounded under NA$_1$ (see Proposition \ref{prop:NA1}), there exists a bounded polyhedral set $\cP\subset{\rm span}(\strat\cap\cL)$  such that $\strat\cap\cL\subseteq\cP$ (see, e.g., \cite[Theorem 20.4]{Rockafellar}).
Denote by $\{p_1,\ldots,p_N\}$ the set of extreme points of $\cP$. Since a linear function defined on a polyhedral set attains its maximum on the set of extreme points, it holds that
\[
\langle\pi,R\rangle 
\leq \max_{j=1,\ldots,N}\langle p_j,R\rangle,
\qquad\text{ for all }\pi\in\strat\cap\cL.
\]
By monotonicity of $U$, this implies that
\[
U^+(1+\langle\pi,R\rangle)
\leq \sum_{j=1}^NU^+(1+\langle p_j,R\rangle) =: \zeta,
\qquad\text{ for all }\pi\in\strat\cap\cL.
\]
We proceed to show that $\EE[\zeta]<+\infty$. Since $U(1)$ is bounded from below, we can assume without loss of generality that $U(1)\geq0$. We recall from \cite[Lemma 2.2]{RasonyStettner06} the inequality 
\be	\label{eq:ineq}
U^+(\lambda x)\leq 2\lambda\bigl(U^+(x)+U(2)\bigr),
\qquad\text{ for all $x>0$ and $\lambda\geq1$.}
\ee
Let $\phi$ be an element of the relative interior of $\strat\cap\cL$ and $\varepsilon_j\in(0,1]$ such that $\phi+\varepsilon_j(p_j-\phi)\in\strat\cap\cL$, for all $j=1,\ldots,N$.
By inequality \eqref{eq:ineq} and monotonicity of $U$, together with the fact that $\phi\in\strat\cap\cL\subseteq\strat_{\rm adm}$, we obtain 
\be	\label{eq:bound}	\ba
U^+\bigl(1+\langle p_j,R\rangle\bigr)
&= U^+\bigl(1+\langle\phi,R\rangle+\langle p_j-\phi,R\rangle\bigr)	\\
&\leq \frac{2}{\varepsilon_j}\left(U^+\bigl(\varepsilon_j(1+\langle\phi,R\rangle)+\varepsilon_j\langle p_j-\phi,R\rangle\bigr)+U(2)\right)	\\
&\leq  \frac{2}{\varepsilon_j}\left(U^+\bigl(1+\langle\phi+\varepsilon_j(p_j-\phi),R\rangle\bigr)+U(2)\right).
\ea\ee
Due to the assumption that $\sup_{\pi\in\strat}\EE[U^+(V^{\pi}_1)]<+\infty$, the first term on the last line of \eqref{eq:bound} is integrable, for each $j=1,\ldots,N$. The same assumption implies that $\EE[U(1)]<+\infty$, from which $\EE[U(2)]<+\infty$ follows by concavity of $U$. This proves that the random variable $\zeta$ is integrable. 
Let now $(\pi^n)_{n\in\N}$ be a sequence in $\strat\cap\cL$ converging to some element $\pi^0\in\strat\cap\cL$. An application of Fatou's lemma, together with the continuity of $U$, yields that
\[
\limsup_{n\rightarrow+\infty}u(\pi^n)
\leq \EE\bigl[\limsup_{n\rightarrow+\infty}U(1+\langle\pi^n,R\rangle)\bigr]
= u(\pi^0),
\]
thus proving the upper semi-continuity of the function $u$.\\
$(ii)\Rightarrow(i)$: by way of contradiction, let $\pi^*\in\strat\cap\cL$ be the maximizer in \eqref{eq:opt_utility} and suppose that NA$_1$ fails to hold. By Proposition \ref{prop:NA1}, there exists $\theta\in\cI_{\rm arb}\cap\strathat$. It holds that $\pi^*+\theta\in\strat\cap\cL$ and $\EE[U(V^{\pi^*+\theta}_1)] > \EE[U(V^{\pi^*}_1)]$, thus contradicting the optimality of $\pi^*$.
\end{proof}

The above theorem asserts the equivalence between NA$_1$ and {\em market viability}, intended as the existence of an optimal strategy for every well-posed expected utility maximization problem.
In particular, the proof makes clear that one of the crucial consequences of NA$_1$ is the compactness of the set $\strat\cap\cL$ of non-redundant allowed strategies (see Proposition \ref{prop:NA1}).

\begin{rem}
The proof of Theorem \ref{thm:utility} relies on the fact that, under NA$_1$, the set $\strat\cap\cL$ and the function $u$ have no common directions of recession. The relevance of this property in expected utility maximization problems has been first recognized in the early work \cite{Bertsekas74}.
\end{rem}

\begin{rem}	\label{rem:ElsingerSummer}
In discrete-time models without trading constraints, it is well-known that market viability is equivalent to no classical arbitrage (see \cite{RasonyStettner05,RasonyStettner06} and \cite[Theorem 3.3]{FollmerSchied}). In view of Remark \ref{rem:compare_arb}, the same holds true in the case of conic  constraints. 
In the case of convex non-conic constraints, Theorem \ref{thm:utility} shows that market viability is equivalent to the weaker NA$_1$ condition. From an economic standpoint, this result implies that assuming no classical arbitrage may pose unnecessary restrictions on the model.
In the special case of a finite probability space, this insight already appeared in \cite{ElsingerSummer01}, where the authors proved the equivalence between the existence of a solution to optimal portfolio problems and the validity of a condition called by the authors {\em no unlimited arbitrage}.
Translated in our context, no unlimited arbitrage corresponds to the existence of a strategy $\theta\in\Theta$ such that there are no arbitrage opportunities relative to $\theta$, in the sense of Remark \ref{rem:rel_arb}. As shown in Remark \ref{rem:NA1_relarb} below, this is equivalent to NA$_1$ and, therefore, the results of \cite{ElsingerSummer01} can be recovered as a special case.\footnote{The finiteness condition appearing in part $(ii)$ of Theorem \ref{thm:utility} is always satisfied on a finite probability space.}
In continuous-time semimartingale models, the connection between NA$_1$ and the solvability of expected utility maximization problems is discussed in \cite{KaratzasKardaras07} and \cite{ChoulliDengMa15} (in an It\^o process setting, earlier results in this direction have been obtained in \cite{LoewensteinWillard00}).
\end{rem}

The NA$_1$ condition admits an equivalent characterization in terms of the existence of a (supermartingale) {\em deflator} or of a {\em num\'eraire portfolio}, defined as follows.

\begin{defn}	\label{def:numeraire}
A random variable $Z\in L^0_+$ with $P(Z>0)=1$ is said to be a {\em deflator} if
\be	\label{eq:deflator}
\EE[Z \,V^{\pi}_1]\leq 1,
\qquad\text{ for all }\pi\in\strat.
\ee
The set of all deflators is denoted by $\cD$.\\
An allowed trading strategy $\rho\in\strat$ is said to be a {\em num\'eraire portfolio} if $1/V^{\rho}_1\in\cD$, meaning that
\be	\label{eq:numeraire}
\EE[V^{\pi}_1/V^{\rho}_1]\leq 1,
\qquad\text{ for all }\pi\in\strat.
\ee
\end{defn}

It is well-known (see, e.g., \cite{Becherer01}) that a num\'eraire portfolio is unique in the sense that if $\rho^1$ and $\rho^2$ satisfy \eqref{eq:numeraire}, then $\rho^1-\rho^2\in\cL^{\perp}$. The num\'eraire portfolio is therefore uniquely defined on $\strat\cap\cL$.
The next theorem shows that NA$_1$ is necessary and sufficient for the existence of the num\'eraire portfolio. In a general semimartingale setting, the corresponding result has been proved in \cite{KaratzasKardaras07}. 
In the present context, Theorem \ref{thm:utility} enables us to give a short and simple proof based on log-utility maximization, thus highlighting the central role of market viability.
Besides simplifying  the techniques employed in \cite[Lemma 6.2 and Theorem 6.3]{KS09}, our proof can be easily generalized to the multi-period setting, as will be shown in Section \ref{sec:multiperiod}.

\begin{thm}	\label{thm:numeraire}
The following are equivalent:
\begin{enumerate}[(i)]
\item the {\em NA}$_1$ condition holds;
\item $\cD\neq\emptyset$;
\item there exists the num\'eraire portfolio.
\end{enumerate}
Moreover, $\rho\in\strat$ is the num\'eraire portfolio if and only if it is {\em relatively log-optimal}, in the sense that it satisfies
$\EE[\log(V^{\pi}_1/V^{\rho}_1)] \leq 0$, for all $\pi\in\strat$.
\end{thm}
\begin{proof}
$(i)\Rightarrow(iii)$: as a preliminary, similarly as in \cite{Kardaras09,KS09}, let $(f_n)_n$ be a family of  functions such that $f_n:\R^d\rightarrow(0,1]$ and $\EE[\log(1+\|R\|)f_n(R)]<+\infty$, for each $n\in\N$, and $f_n\nearrow1$ as $n\rightarrow+\infty$. A specific choice is for instance given by $f_n(x)=\ind_{\{\|x\|\leq 1\}}+\ind_{\{\|x\|>1\}}\|x\|^{-1/n}$. For each $n\in\N$, define the function $(\omega,x)\mapsto U_n(\omega,x):=\log(x)f_n(R(\omega))$, for $(\omega,x)\in\Omega\times(0,+\infty)$, with $U_n(\omega,0):=\lim_{x\downarrow0}U_n(\omega,x)=-\infty$. For each $n\in\N$, it holds that $U_n\in\cU$ and
\be	\label{eq:ineq_log}
\EE\bigl[U^+_n(1+\langle\pi,R\rangle)\bigr]
\leq \|\pi\| + \EE\bigl[\log(1+\|R\|)f_n(R)\bigr] < +\infty,
\qquad\text{ for all }\pi\in\strat.
\ee
If NA$_1$ holds, Proposition \ref{prop:NA1} implies that $\strat\cap\cL$ is bounded and, therefore, it holds that $\sup_{\pi\in\strat}\EE[U^+_n(1+\langle\pi,R\rangle)]<+\infty$. 
For each $n\in\N$, Theorem \ref{thm:utility} gives then the existence of an element $\rho^n\in\strat\cap\cL$ which is the maximizer in  \eqref{eq:opt_utility} for $U=U_n$.
For an arbitrary element $\pi\in\strat$ and $\varepsilon\in(0,1)$, let $\pi^{\varepsilon}:=\varepsilon\pi+(1-\varepsilon)\rho^n\in\strat$.
The optimality of $\rho^n$ together with the elementary inequality $\log(x)\geq(x-1)/x$, for $x>0$, implies that
\begin{align}
0 &\geq \frac{1}{\varepsilon}\Bigl(\EE\bigl[U_n(1+\langle\pi^{\varepsilon},R\rangle)\bigr]-\EE\bigl[U_n(1+\langle\rho^n,R\rangle)\bigr]\Bigr)	\notag\\
&= \frac{1}{\varepsilon}\EE\bigl[\log(V^{\pi^{\varepsilon}}_1/V^{\rho^n}_1)f_n(R)\bigr]	
\geq \EE\left[\frac{\langle\pi-\rho^n,R\rangle}{1+\langle\rho^n,R\rangle+\varepsilon\langle\pi-\rho^n,R\rangle}f_n(R)\right].
\label{eq:rel_log_opt}
\end{align}
Noting that $\frac{x}{y+\varepsilon x}\geq\frac{x}{y+x/2}\geq-2$, for all $\varepsilon\in(0,1/2)$, $y>0$ and $x\geq-y$, we can let $\varepsilon\searrow0$ and apply Fatou's lemma, thus obtaining
\be	\label{eq:rel_log_opt_2}
\EE\left[\frac{\langle\pi-\rho^n,R\rangle}{1+\langle\rho^n,R\rangle}f_n(R)\right] \leq 0,
\qquad\text{ for all }\pi\in\strat\text{ and }n\in\N.
\ee
Since $\strat\cap\cL$ is compact (see Proposition \ref{prop:NA1}), we may assume that the sequence $(\rho^n)_{n\in\N}$ converges to some $\rho\in\strat\cap\cL$ as $n\rightarrow+\infty$. 
Therefore, since $\langle\pi-\rho^n,R\rangle/(1+\langle\rho^n,R\rangle)\geq-1$ a.s. and recalling that $f_n\nearrow1$ as $n\rightarrow+\infty$, another application of Fatou's lemma gives that
\[
\EE\left[\frac{\langle\pi-\rho,R\rangle}{1+\langle\rho,R\rangle}\right] \leq 0,
\qquad\text{ for all }\pi\in\strat.
\]
Equivalently, it holds that $\EE[V^{\pi}_1/V^{\rho}_1]\leq1$, for all $\pi\in\strat$. 
In view of Definition \ref{def:numeraire}, we have thus shown that NA$_1$ implies the existence of the num\'eraire portfolio.\\
$(iii)\Rightarrow(ii)$: this implication is immediate by Definition \ref{def:numeraire}.\\
$(ii)\Rightarrow(i)$: 
let $Z\in\cD$ and consider a random variable $\xi\in L^0_+$ with $P(\xi>0)>0$ such that, for every $n\in\N$, there exists $\pi^n\in\strat$ such that $V^{\pi^n}_1(1/n)\geq\xi$ a.s. 
Definition \eqref{def:numeraire} implies that
\[
\EE[Z\,\xi] \leq \EE\bigl[Z\,V^{\pi^n}_1(1/n)\bigr] = \frac{1}{n}\EE\bigl[Z\,V^{\pi^n}_1\bigr] \leq \frac{1}{n},
\qquad\text{ for all }n\in\N.
\]
Since $Z>0$ a.s., letting $n\rightarrow+\infty$ yields that $\xi=0$ a.s., thus proving the validity of NA$_1$.\\
It remains to prove the last assertion of the theorem. If $\rho\in\strat$ satisfies \eqref{eq:numeraire}, then its relative log-optimality is a direct consequence of Jensen's inequality. Conversely, if $\rho\in\strat$ is relatively log-optimal, then \eqref{eq:numeraire} follows by the same arguments used in \eqref{eq:rel_log_opt}-\eqref{eq:rel_log_opt_2}.
\end{proof}

\begin{rem}	\label{rem:log_optimal}
If there exists a {\em log-optimal portfolio}, i.e., an allowed strategy $\rho\in\strat$ satisfying $\EE[\log(V^{\pi}_1)]\leq\EE[\log(V^{\rho}_1)]<+\infty$, for all $\pi\in\strat$, then $\rho$ is also relatively log-optimal and, therefore, coincides with the num\'eraire portfolio. 
The num\'eraire property of the log-optimal portfolio can also be directly deduced from the proof of Theorem \ref{thm:numeraire}. 
In applications, computing the log-optimal portfolio typically represents a simple way to determine the num\'eraire portfolio (see for instance Examples \ref{ex:1dim} and \ref{ex:third_example}).
\end{rem}

\begin{rem}	\label{rem:NA1_relarb}
NA$_1$ is equivalent to the existence of a strategy $\theta\in\strat$ with $V^{\theta}_1>0$ a.s. such that there are no arbitrage opportunities relative to $\theta$, in the sense of Remark \ref{rem:rel_arb}. Indeed, suppose there exists $\theta\in\strat$ with $V^{\theta}_1>0$ a.s. and let $\pi\in\strathat$. Then $\pi+\theta$ is an arbitrage opportunity relative to $\theta$ if and only if $\pi\in\cI_{\rm arb}$.
Conversely, if NA$_1$ holds, then there do not exist arbitrage opportunities relative to the num\'eraire portfolio $\rho$, as a consequence of \eqref{eq:numeraire}.
However, absence of arbitrage opportunities relative to some strategy $\theta\in\strat$ with $V^{\theta}_1>0$ a.s. does not suffice to conclude that $\theta$ is the num\'eraire portfolio (see Example \ref{ex:second_example} for an explicit counterexample).
\end{rem}

Theorems \ref{thm:utility} and \ref{thm:numeraire} represent the central results of arbitrage theory based on NA$_1$. 
For completeness, we now state the fundamental theorem of asset pricing based on no classical arbitrage, in the general version of \cite[Theorem 4]{Rokhlin05} specialised for a one-period setting.
We give a simple proof inspired by \cite[Proposition 2.1.5]{KabSaf} and \cite[Theorem 3.7]{Kardaras09}, which in turn follow an original idea of \cite{Rogers94}.
Similarly to Theorem \ref{thm:numeraire}, the proof is based on utility maximization arguments.
For a set $A\subseteq\R^d$, we  denote by ${\rm cone}\,A$ its conic hull. 

\begin{thm}	\label{thm:arb}
Suppose that the set ${\rm cone}\,\strat$ is closed. Then  no classical arbitrage holds if and only if there exists a probability measure $Q\sim P$ such that $\EE^Q[V^{\pi}_1]\leq 1$, for all $\pi\in{\rm cone}\,\strat$. 
\end{thm}
\begin{proof}
Observe first that $\cI_{\rm arb}\cap\strat=\emptyset$ if and only if $\cI_{\rm arb}\cap({\rm cone}\,\strat)=\emptyset$. In turn, this implies that no classical arbitrage holds if and only if $\cI_{\rm arb}\cap C=\emptyset$, where $C:=({\rm cone}\,\strat)\cap\cL$. Define the proper convex function $f:C\ni\pi\mapsto f(\pi):=\EE'[\exp(-1-\langle\pi,R\rangle)]$, where $\EE'$ denotes expectation with respect to the probability measure $\PP'$ defined by $\ud\PP'/\ud\PP=e^{-\|R\|^2}/\EE[e^{-\|R\|^2}]$. 
By Fatou's lemma, the function $f$ is lower semi-continuous. 
Since $C$ is closed by assumption, \cite[Theorem 27.3]{Rockafellar} implies that the function $f$ admits a minimizer $\pi^*\in C$ if it has no directions of recession in common with the cone $C$. 
By \cite[Theorem 8.5]{Rockafellar}, this amounts to verifying that 
\be	\label{eq:recession}
\hat{f}(\pi) := \lim_{\gamma\rightarrow+\infty}\frac{f(\gamma\pi)}{\gamma}>0,
\qquad\text{ for all }\pi\in C\setminus\{0\}.
\ee
We now show that \eqref{eq:recession} is always satisfied under no classical arbitrage. Arguing by contradiction, let $\pi\in C\setminus\{0\}$ such that $\hat{f}(\pi)\leq0$. In this case, by Fatou's lemma,  it holds that
\begin{align*}
0 &\geq \hat{f}(\pi) 
\geq \EE'\left[\liminf_{\gamma\rightarrow+\infty}\frac{e^{-1-\gamma\langle\pi,R\rangle}}{\gamma}\right]
\geq \EE'\left[\liminf_{\gamma\rightarrow+\infty}\frac{e^{-1-\gamma\langle\pi,R\rangle}}{\gamma}\ind_{\{\langle\pi,R\rangle<0\}}\right].
\end{align*}
This implies that necessarily $\langle\pi,R\rangle\geq0$ a.s. Since $\pi\in\cL$, this contradicts no classical arbitrage.
\cite[Theorem 27.3]{Rockafellar} then yields the existence of an element $\pi^*\in C$ such that $f(\pi^*)\leq f(\pi)$, for all $\pi\in C$.
The definition of $\PP'$ implies that differentiation and integration can be interchanged, so that the gradient of the function $f$ at $\pi^*$ is given by $\nabla f(\pi^*)=-\EE'[\exp(-1-\langle\pi^*,R\rangle)R]$.
Therefore, since $C$ is a cone and $f$ is finite on $C$, \cite[Theorem 27.4]{Rockafellar}  implies that
\[
0 \geq \bigl\langle\pi,-\nabla f(\pi^*)\bigr\rangle
= \EE'\bigl[e^{-1-\langle\pi^*,R\rangle}\langle\pi,R\rangle\bigr].
\]
Setting $\ud Q/\ud P=e^{-V^{\pi^*}_1-\|R\|^2}/\EE[e^{-V^{\pi^*}_1-\|R\|^2}]$ yields a probability measure $Q\sim P$ such that $\EE^Q[V^{\pi}_1]\leq1$, for all $\pi\in C$ and, hence, for all $\pi\in{\rm cone}\,\strat$.

Conversely, suppose there exists a probability measure $Q\sim P$ such that $\EE^{Q}[V^{\pi}_1]\leq 1$, for all $\pi\in{\rm cone}\,\strat$. Then, for every $\pi\in\strat$, it holds that $\EE^{Q}[\langle\pi,R\rangle]\leq0$. If $\pi\in\cI_{\rm arb}\cap\strat$, this implies that $\langle\pi,R\rangle\leq0$ $Q$-a.s. However, since $Q\sim P$, this contradicts the fact that $\pi\in\cI_{\rm arb}$.
\end{proof}

\begin{rem}
Theorem \ref{thm:arb} does not hold without the assumption of closedness of ${\rm cone}\,\strat$.\footnote{The same assumption is required in the fundamental theorem of asset pricing in the formulation of \cite{CPT01}.
\cite[Theorem 4]{Rokhlin05} requires the closedness of $\proj_{\cL}({\rm cone}\,\strat)$, the set of all vectors in $\R^d$ that are projections onto $\cL$ of elements of ${\rm cone}\,\strat$. 
In our setting, since $\cL^{\perp}\subseteq\strathat$, it holds that $\proj_{\cL}({\rm cone}\,\strat)=({\rm cone}\,\strat)\cap\cL$.
This implies that $\proj_{\cL}({\rm cone}\,\strat)$ is closed if and only if ${\rm cone}\,\strat$ is closed.}
Indeed, one can construct a counterexample along the lines of \cite[Example 1]{Rokhlin05} where no classical arbitrage holds but there does not exist a probability measure $Q\sim P$ such that $\EE^Q[V^{\pi}_1]\leq1$, for all $\pi\in{\rm cone}\,\strat$.
Observe that, in comparison to no classical arbitrage, NA$_1$ has the additional advantage of not requiring any extra technical condition on the model.
\end{rem}

The probability measure $Q$ appearing in Theorem \ref{thm:arb} represents an {\em equivalent supermartingale measure} (ESMM). If NA$_1$ holds and the num\'eraire portfolio $\rho$ satisfies $\EE[1/V^{\rho}_1]=1$, then an ESMM $Q$ can be defined by setting $\ud Q/\ud P=1/V^{\rho}_1$. However, this is not always possible, even when ${\rm cone}\,\strat$ is closed and no classical arbitrage holds, as the following simple example illustrates (see also \cite[Example 6]{Becherer01} for a related example in an unconstrained setting).

\begin{example}	\label{ex:1dim}
Let $d=1$ and suppose that $R=e^Y-1$, with $Y\sim\mathcal{N}(0,1)$. In this case, it holds that $\cS=[-1,+\infty)$ and $\strat_{\rm adm}=[0,1]$ (i.e., short-selling and borrowing from the riskless asset are prohibited). Suppose that $\strat_{\rm c}=[0,c]$, for some $c\in[0,1]$, so that $\strat=[0,c]$. Clearly, no classical arbitrage holds and, therefore, there exists an ESMM $Q$. For instance, it can be easily checked that $\ud Q/\ud P = \exp(\alpha Y-\alpha^2/2)$ defines an ESMM, for any $\alpha\leq-1/2$.
However, if $c<1/2$, the num\'eraire portfolio $\rho$ cannot be used to construct an ESMM, since $\EE[1/V^{\rho}_1]<1$. Indeed, it can be easily checked that the function $h:[0,1]\rightarrow\R$ defined by $h(\pi):=\EE[\log(V^{\pi}_1)]$ is finite-valued, strictly concave and achieves its maximum at $1/2$, so that $h'(\pi)>0$ for all $\pi<1/2$. Therefore, if $c<1/2$, the log-optimal portfolio and, therefore, the num\'eraire portfolio $\rho$ (see Remark \ref{rem:log_optimal}) are given by $\rho=c$ and it holds that $h'(\rho)>0$ or, equivalently, $\EE[1/V^{\rho}_1]<1$.
\end{example}

\subsection{Hedging and valuation of contingent claims}	\label{sec:hedging}

The pricing of contingent claims is traditionally based on the paradigm of no classical arbitrage. In this section, we show that the weaker NA$_1$ condition suffices to develop a general and effective theory for the hedging and valuation of contingent claims in the presence of convex constraints.
We first prove the fundamental super-hedging duality. 
Recall that for a random variable $\xi\in L^0_+$ (contingent claim) its super-hedging value $v(\xi)$ is defined as in \eqref{eq:superhedge_price}, with the usual convention $\inf\emptyset=+\infty$.

\begin{thm}\label{thm:duality}
Suppose that {\em NA}$_1$ holds and let $\xi\in L^0_+$. Then 
\be	\label{eq:duality}
v(\xi) = \sup_{Z\in\cD}\EE[Z\xi].
\ee
Moreover, there exists a pair $(v,\pi)\in\R_+\times\strat$ such that $\xi=V^{\pi}_1(v)$ a.s. and $\EE[Z\xi]=v$, for some $Z\in\cD$, 
if and only if there exists 
an element $Z^*\in\cD$ such that $\EE[Z^*\xi]=\sup_{Z\in\cD}\EE[Z\xi]<+\infty$.
\end{thm}
\begin{proof}
Let $\cV(\xi):=\{v>0 : \exists\;\pi\in\strat \text{ such that }vV_1^{\pi}\geq\xi\text{ a.s.}\}$ and $\cC:=\{V^{\pi}_1 : \pi\in\strat\cap\cL\}-L^0_+$.
If $v\in\cV(\xi)$, there exists $\pi\in\strat$ such that $vV^{\pi}_1\geq\xi$ a.s. Then, for every $Z\in\cD$ it holds that
\[
\EE[Z\xi] \leq v\,\EE[Z V^{\pi}_1] \leq v.
\]
By taking the supremum over all $Z\in\cD$ and the infimum over all $v\in\cV(\xi)$, we obtain that $v(\xi)\geq\sup_{Z\in\cD}\EE[Z\xi]=:v^*$. 
The converse inequality is trivial if $v^*=+\infty$. Assuming therefore that $0<v^*<+\infty$, we will show that $v(\xi)>v^*$ cannot hold.
Indeed, if $v(\xi)>v^*$, then $\xi\notin v^*\cC$. 
Let $\rho$ be the num\'eraire portfolio (which exists by Theorem \ref{thm:numeraire}). Being closed in $L^0$ (see Lemma \ref{lem:closedness} below) and bounded in $L^1$, the set $v^*\cC/V^{\rho}_1$ is closed in $L^1$. 
Therefore, by the Hahn-Banach theorem (see, e.g., \cite[Theorem A.58]{FollmerSchied}), there exists  a bounded random variable $\bar{Z}$ such that
\be	\label{eq:HB}
+\infty > \frac{1}{v^*}\EE\left[\bar{Z}\frac{\xi}{V^{\rho}_1}\right] > \sup_{X\in\cC}\EE\left[\bar{Z} \frac{X}{V^{\rho}_1}\right] =: s.
\ee
Since $-n\ind_{\{\bar{Z}<0\}}\in\cC$, for all $n\geq0$, inequality \eqref{eq:HB} implies that $\bar{Z}\geq0$ a.s. and $P(\bar{Z}>0)>0$. Moreover, since $1\in\cC$, it holds that $s>0$. 
For $\varepsilon\in(0,1)$, we define
\be	\label{eq:Zeps}
Z^{\varepsilon} := \left(\varepsilon + (1-\varepsilon)\frac{\bar{Z}}{s}\right)\frac{1}{V^{\rho}_1}.
\ee
It holds that $P(Z^{\varepsilon}>0)=1$ and, for every $\pi\in\strat$,
\[
\EE[Z^{\varepsilon}V^{\pi}_1]
= \varepsilon\EE\left[\frac{V^{\pi}_1}{V^{\rho}_1}\right] + \frac{1-\varepsilon}{s}\EE\left[\bar{Z}\frac{V^{\pi}_1}{V^{\rho}_1}\right] \leq 1,
\]
thus showing that $Z^{\varepsilon}\in\cD$, for all $\varepsilon\in(0,1)$. Moreover, for a sufficiently small $\varepsilon$,  \eqref{eq:HB} together with \eqref{eq:Zeps} implies that $\EE[Z^{\varepsilon}\xi]>v^*=\sup_{Z\in\cD}\EE[Z\xi]$, which is absurd. Therefore, we must have $\xi\in v^*\cC$ , thus proving that $v(\xi)\leq v^*=\sup_{Z\in\cD}\EE[Z\xi]$.\\
To prove the last assertion of the theorem, observe that the first part of the proof yields that $v^*V^{\pi}_1\geq\xi$ a.s., for some $\pi\in\strat$. If there exists $Z^*\in\cD$ such that $v^*=\EE[Z^*\xi]$, then we have that
\[
v^* = \EE[Z^*\xi] \leq v^*\EE[Z^*V^{\pi}_1] \leq v^*.
\]
Since $Z^*>0$ a.s., this implies that $\xi=V^{\pi}_1(v^*)$ a.s.
Conversely, if $\xi=V^{\pi}_1(v)$ a.s. for some $(v,\pi)\in\R_+\times\strat$ with $v=\EE[Z^*\xi]$, for some $Z^*\in\cD$, then \eqref{eq:deflator} implies that $\EE[Z^*\xi]=\sup_{Z\in\cD}\EE[Z\xi]$.
\end{proof}

\begin{lem}	\label{lem:closedness}
If {\em NA}$_1$ holds, then the set $\cC:=\{V^{\pi}_1 :\pi\in\strat\cap\cL\}-L^0_+$ is closed in $L^0$.
\end{lem}
\begin{proof}
Let $(X_n)_{n\in\N}\subseteq\cC$ be a sequence converging in $L^0$ to a random variable $X$ as $n\rightarrow+\infty$. For each $n\in\N$, it holds that $X_n=1+\langle\pi_n,R\rangle-A_n$, for $(\pi_n,A_n)\in(\strat\cap\cL)\times L^0_+$. By Proposition \ref{prop:NA1}, NA$_1$ implies that the set $\strat\cap\cL$ is compact and, therefore, there exists a subsequence $(\pi_{n_m})_{m\in\N}$ converging to an element $\pi\in\strat\cap\cL$. In turn, this implies that the sequence $(A_{n_m})_{m\in\N}$ converges in probability to a random variable $A\in L^0_+$, thus proving the closedness of $\cC$ in $L^0$.
\end{proof}

Whenever the quantity $\sup_{Z\in\cD}\EE[Z\xi]$ is finite, it provides the super-hedging value of $\xi$. 
In a general semimartingale setting, the duality relation \eqref{eq:duality} has been stated in \cite[Section 4.7]{KaratzasKardaras07}. We contribute by providing a transparent and self-contained proof in a one-period setting. In addition, Theorem \ref{thm:duality} provides a necessary and sufficient condition for the  attainability of a contingent claim $\xi$.
When perfect hedging is not possible, one may resort to several alternative hedging approaches, which are all feasible under NA$_1$ even if no classical arbitrage fails to hold.
A first possibility is represented by {\em hedging with minimal shortfall risk}, corresponding to
\be	\label{eq:shortfall}
\EE\bigl[\ell(\xi-vV^{\pi}_1)\bigr] = \min\,!
\qquad\text{ over all }(v,\pi)\in(0,v_0]\times\strat,
\ee
for some initial capital $v_0>0$, where $\ell:\R\rightarrow\R$ is an increasing convex loss function such that $\ell(x)=0$, for all $x\leq0$, and $\EE[\ell(\xi)]<+\infty$ (see \cite[Section 8.2]{FollmerSchied}). 
Problem \eqref{eq:shortfall} can be solved by first minimizing $\EE[\ell(\xi-Y)]$ over all random variables $Y\in L^0_+$ such that $\sup_{Z\in\cD}\EE[ZY]\leq v_0$ and then considering the pair $(v(Y^*),\pi^*)$ which super-replicates the minimizing random variable $Y^*$ (if $\ell$ is strictly increasing on $[0,+\infty)$, then $v(Y^*)=v_0$). As long as NA$_1$ holds, the feasibility of this approach is ensured by Theorem \ref{thm:duality}.

An alternative way to hedge and compute the value of a contingent claim $\xi$ is provided by {\em utility indifference valuation}. For a given utility function $u$ and an initial capital $v>0$, this corresponds to finding the solution $p=p(\xi)$ to the equation
\be	\label{eq:uip}
\sup_{\pi\in\strat}\EE\bigl[u(vV^{\pi}_1)\bigr] =
\sup_{\pi\in\strat}\EE\bigl[u((v-p)V^{\pi}_1+\xi)\bigr].
\ee
Defining $U_p^{\eta}(x,\omega):=u((v-\eta p)x+\eta\xi(\omega))$, for $\eta\in\{0,1\}$, Theorem \ref{thm:utility} with $U=U^{\eta}_p$ shows that NA$_1$ is sufficient for the solvability of the two maximization problems appearing in \eqref{eq:uip}. Whenever it exists, $p(\xi)$ represents a (buyer) value for $\xi$, while the strategy $\pi^*$ that achieves the supremum on the right-hand side of \eqref{eq:uip} with $p=p(\xi)$ provides a hedging strategy for $\xi$.

As a variant of the latter approach, one can consider {\em marginal} utility indifference valuation, in the sense of \cite{Davis97}. This corresponds to finding the value $p=p'(\xi)$ which solves
\[
\lim_{\eta\downarrow0}\frac{\EE[U_p^{\eta}(V^{\pi^*}_1)]-\EE[U_p^0(V^{\pi^*}_1)]}{\eta} = 0.
\]
where $U^{\eta}$ is defined as above, for $\eta\in[0,1]$, and $\pi^*\in\strat$ is the strategy solving problem \eqref{eq:opt_utility} with $U=u$. Similarly as in \cite{FontanaRunggaldier13}, if NA$_1$ holds and $u(x)=\log(x)$, it can be shown that
\be	\label{eq:rw_price}
p'(\xi) = \EE[\xi/V^{\rho}_1],
\ee
as long as the expectation is finite, where $\rho$ denotes the num\'eraire portfolio (see Theorem \ref{thm:numeraire}). 
In the context of the Benchmark Approach (see \cite{BuhlmannPlaten03,PlatenHeath}), formula \eqref{eq:rw_price} corresponds to the well-known {\em real-world pricing formula}, which is applicable as long as NA$_1$ is satisfied.

\section{Factor models with arbitrage under borrowing constraints}	\label{sec:arbitrage_model}

In this section, we study the arbitrage concepts discussed above in the context of a one-period factor model, under constraints on the fraction of wealth that can be borrowed/invested on the riskless asset. We start from a general model and then consider more specific cases.

\subsection{A general factor model}	\label{sec:gen_factor}
In the setting of Section \ref{sec:1p}, we assume that asset returns are generated by the factor model
\be	\label{eq:factors}
R = QY,
\ee
where $Q\in\R^{d\times\ell}$ and $Y=(Y_1,\ldots,Y_\ell)^{\top}$ is an $\ell$-dimensional random vector with independent components, for some $\ell\in\N$. A non-diagonal matrix $Q$ permits to introduce general correlation structures among the $d$ asset returns. 
Without loss of generality, we assume that ${\rm rank}(Q)=d$. Under this assumption, it holds that $\cL^{\perp}=\{0\}$.

\begin{rem}	\label{rem:factor_models}
Multi-factor models are widely employed in financial economics and econometrics, the {\em Arbitrage Pricing Theory} of \cite{Ross76} and its extensions representing some of the most notable instances (see \cite[Chapter 5]{BF17}, \cite{ConnorKorajczyk95} and \cite[Chapter 6]{CLMK} for overviews on the topic).
Multi-factor asset pricing models can always be written in the form \eqref{eq:factors}, modulo the assumption of independent factors.\footnote{Under this assumption, representation \eqref{eq:factors} enables us to reduce the analysis to $\ell$ independent sources of randomness. We stress that any correlation structure among the asset returns $R$ can be generated by a suitable specification of the matrix $Q$.}
To this effect, recall first that the random vector $R$ represents the excess returns of $d$ risky assets with respect to a baseline security, usually chosen as a riskless asset. 
Multi-factor models are typically stated in the form
\be	\label{eq:factor_model}
R = \EE[R] + BF + \epsilon,
\ee
where $\EE[R]$ is the vector of risk premia, $F$ is a $k$-dimensional random vector of common risk factors, for some $k<d$, $B\in\R^{d\times k}$ is the matrix of factor loadings and $\epsilon$ is a $d$-dimensional random vector of idiosyncratic (asset-specific) risk factors. 
Depending on the modeling choices, $F$ can represent a vector of economic factors or statistical factors.
In the standard formulation (see, e.g., \cite[Chapter 7]{Ingersoll87}), all components of $F$ and $\epsilon$ are assumed to be uncorrelated.
Notice now that factor model \eqref{eq:factor_model} can be written in the form \eqref{eq:factors} by setting $Q=(\EE[R],B,\mathbb{I}_d)$ and $Y=(1,F^{\top},\epsilon^{\top})^{\top}$, where $\mathbb{I}_d$ denotes the $(d\times d)$ identity matrix.
In the special case of absence of idiosyncratic risk, the vector $Y$ can be directly identified with $F$.
Equation \eqref{eq:factors} therefore provides the simplest unifying representation of multi-factor asset pricing models. 
\end{rem}

For $k=1,\ldots,\ell$, we denote by $\cY_k$ the support of $Y_k$ and let $y_k^{\inf}:=\inf\cY_k$ and $y^{\sup}_k:=\sup\cY_k$. 
In this section, we work under the following standing assumption:
\be	\label{eq:support}
y_1^{\inf}=0,
\quad y_1^{\sup}=+\infty
\qquad\text{ and }\qquad
y_k^{\inf}<0<y_k^{\sup},
\quad\text{ for all }k=2,\ldots,\ell.
\ee
As will become clear in the sequel, condition \eqref{eq:support} corresponds to viewing the first factor $Y_1$ as the driving force of possible arbitrage opportunities, while the remaining factors cannot be exploited to generate arbitrage.\footnote{The only requirement in order to allow for arbitrage opportunities is the existence of a linear combination of factors with positive support. The assumption that $Y^1$ has positive support is only made for convention.}
In the context of the factor model \eqref{eq:factors}-\eqref{eq:support}, the following lemma gives a necessary and sufficient condition to ensure positive asset prices.
For $i=1,\ldots,d$ and $k=1,\ldots,\ell$, we denote by $q_{i,k}$ the element on the $i$-th row and $k$-th column of $Q$.

\begin{lem}	\label{lem:pos_prices}
In the context of the model of this section, for each $i=1,\ldots,d$, it holds that $R^i\geq-1$ a.s. if and only if the following condition is satisfied:
\be	\label{eq:pos_prices}
q_{i,1}\geq0
\qquad\text{ and }\qquad
\sum_{k=2}^\ell\bigl(q^+_{i,k}y_k^{\inf}-q^-_{i,k}y^{\sup}_k\bigr) \geq -1,
\ee
with the convention $0\times (-\infty)=0$ and $0\times(+\infty)=0$.
\end{lem}
\begin{proof}
Condition \eqref{eq:pos_prices} is obviously sufficient to ensure that $R^i\geq-1$ a.s., for all $i=1,\ldots,d$. 
Conversely, let $i\in\{1,\ldots,d\}$ and suppose that $R^i\geq-1$ a.s. For all $n\in\N$ and $k=1,\ldots,\ell$, let 
\[
y^{\inf}_k(n) := \Bigl(y^{\inf}_k+\frac{1}{n}\Bigr)\vee(-n)
\qquad\text{ and }\qquad
y^{\sup}_k(n) := \Bigl(y^{\sup}_k-\frac{1}{n}\Bigr)\wedge n,
\]
where $x\vee z:=\max\{x,z\}$ and $x\wedge z:=\min\{x,z\}$, for any $(x,z)\in\R^2$.
With this notation, it holds that $\PP(Y_k \leq y^{\inf}_k(n)) > 0$ and $\PP(Y_k \geq y^{\sup}_k(n)) > 0$, for all $n\in\N$ and $k=1,\ldots,\ell$.
Let $K^+_i:=\{k\in\{1,\ldots,\ell\} : q_{i,k}\geq0\}$ and $K^-_i:=\{1,\ldots,\ell\}\setminus K^+_i$. Since $\sum_{k=1}^{\ell}q_{i,k}Y_k\geq-1$ a.s. and due to the independence of the factors $\{Y_1,\ldots,Y_{\ell}\}$, it holds that
\begin{align*}
0 &< \PP\biggl(Y_{k} \leq y^{\inf}_k(n)\text{ and }Y_{j} \geq y^{\sup}_j(n); \;\forall k\in K^+_i, \forall j\in K^-_i\biggr)	\\
&= \PP\Biggl(\sum_{k\in K^+_i}q_{i,k}Y_{k}\geq-1-\sum_{j\in K^-_i}q_{i,j}Y_j
\;\text{ and }
Y_{k} \leq y^{\inf}_k(n)\text{ and }Y_{j} \geq y^{\sup}_j(n); \;\forall k\in K^+_i, \forall j\in K^-_i\Biggr).
\end{align*}
In turn, this necessarily implies that
$
\sum_{k\in K^+_i}q_{i,k}y^{\inf}_k(n)\geq-1-\sum_{j\in K^-_i}q_{i,j}y^{\sup}_j(n),
$
for each $n\in\N$. Condition \eqref{eq:pos_prices} follows by letting $n\rightarrow+\infty$ and using condition \eqref{eq:support}.
\end{proof}

In particular, condition \eqref{eq:pos_prices} requires that $q_{i,k}\geq0$ if $y^{\sup}_k=+\infty$ and $q_{i,k}\leq0$ if $y^{\inf}_k=-\infty$, for all $i=1,\ldots,d$ and $k=1,\ldots,\ell$.
Observe that condition \eqref{eq:pos_prices} relates the support of the random factors to the dependence structure of the asset returns, represented by the off-diagonal elements of $Q$.
Arguing similarly as in Lemma \ref{lem:pos_prices}, it can be shown that the set $\strat_{\rm adm}$ of admissible strategies can be represented as follows:
\be	\label{eq:theta_adm}
\strat_{\rm adm} 
= \left\{\pi\in\R^d : \pi^{\top}Q_{\bullet,1}\geq0
\text{ and }
\sum_{k=2}^{\ell}\left((\pi^{\top}Q_{\bullet,k})^+y^{\inf}_k-(\pi^{\top}Q_{\bullet,k})^-y^{\sup}_k\right)\geq-1\right\},
\ee
where $Q_{\bullet,k}$ denotes the $k$-th column of the matrix $Q$, with the same convention as in \eqref{eq:pos_prices}.

We now introduce additional trading restrictions, as considered in Section \ref{sec:constraints}. 
More specifically, we assume the presence of {\em borrowing constraints}:
\be	\label{eq:theta_c}
\strat_{\rm c} := \{\pi\in\R^d : \langle\pi,\mathbf{1}\rangle \leq c\},
\ee
for some fixed $c>0$. If $c\in(0,1)$, this corresponds to requiring that at least a proportion $1-c$ of the initial wealth is invested in the riskless asset, while, if $c\geq1$, at most a proportion $c-1$ of the initial wealth can be borrowed from the riskless asset.
Note that, since the set $\strat_{\rm c}$ is not a cone, the notions of arbitrage opportunity and arbitrage of the first kind differ (see Remark \ref{rem:compare_arb}).
As in Section \ref{sec:constraints}, the set $\strat$ of allowed strategies is defined as $\strat:=\strat_{\rm adm}\cap\strat_{\rm c}$.

The following proposition summarizes the arbitrage properties of the factor model under consideration, in the presence of borrowing constraints. We denote by $\cR(Q^{\top})$ the range of the matrix $Q^{\top}$ and by ${\rm e}_k$ the $k$-th vector of the canonical basis of $\R^{\ell}$, for $k=1,\ldots,\ell$.

\begin{prop}	\label{prop:arb_factor}
In the context of the model of this section, the following hold:
\begin{enumerate}[(i)]
\item there are arbitrage opportunities if and only if ${\rm e}_1\in\cR(Q^{\top})$. In that case, it holds that
\be	\label{eq:arb_strategies}
\cI_{\rm arb}\cap\strat 
= \bigl\{\lambda(QQ^{\top})^{-1}Q_{\bullet,1} : \lambda>0\text{ and }\lambda\langle(QQ^{\top})^{-1}Q_{\bullet,1},\mathbf{1}\rangle\leq c\bigr\};
\ee
\item if ${\rm e}_1\in\cR(Q^{\top})$, then ${\rm NA}_1$ holds if and only if $\langle(QQ^{\top})^{-1}Q_{\bullet,1},\mathbf{1}\rangle>0$.
\end{enumerate}
\end{prop}
\begin{proof}
{\em (i)}: 
let $\pi\in\R^d$ such that $\langle\pi,QY\rangle\geq0$ a.s. 
The same argument used to prove Lemma \ref{lem:pos_prices} and representation \eqref{eq:theta_adm} implies that the vector $\pi$ satisfies $\pi^{\top}Q_{\bullet,1}\geq0$ and 
\[
\sum_{k=2}^{\ell}\left((\pi^{\top}Q_{\bullet,k})^+y^{\inf}_k-(\pi^{\top}Q_{\bullet,k})^-y^{\sup}_k\right)\geq0.
\]
Recalling condition \eqref{eq:support}, this implies that $\pi^{\top}Q_{\bullet,k}=0$, for all $k=2,\ldots,\ell$. 
It follows that $\langle\pi,QY\rangle\geq0$ a.s. if and only if $Q^{\top}\pi=\lambda{\rm e}_1$, for some $\lambda\geq0$. 
Since ${\rm rank}(Q)=d$, it holds that $\cI_{\rm arb}=\{\lambda(QQ^{\top})^{-1}Q{\rm e}_1 : \lambda>0\}$, from which representation \eqref{eq:arb_strategies} of the set $\cI_{\rm arb}\cap\strat$ follows directly from the definition of the set $\strat_{\rm c}$ in \eqref{eq:theta_c}.\\
{\em (ii)}:
by Proposition \ref{prop:NA1}, NA$_1$ holds if and only if $\cI_{\rm arb}\cap\strathat=\emptyset$. Representation \eqref{eq:arb_strategies} implies that $\cI_{\rm arb}\cap\strathat=\emptyset$ if and only if $\langle(QQ^{\top})^{-1}Q_{\bullet,1},\mathbf{1}\rangle>0$.
\end{proof}

\begin{rem}
The vector $(QQ^{\top})^{-1}Q_{\bullet,1}$ corresponds to the strategy replicating the factor $Y_1$. While exact replication of $Y_1$ may be precluded by borrowing constraints,  \eqref{eq:arb_strategies} shows that any allowed strategy that replicates a positive fraction of $Y_1$ is an arbitrage opportunity.
The factor $Y_1$ can be (super-)replicated at zero cost if $\langle(QQ^{\top})^{-1}Q_{\bullet,1},\mathbf{1}\rangle\leq0$, in which case NA$_1$ fails.
\end{rem}

\begin{rem}	\label{rem:arbitrage_line}
The proof of Proposition \ref{prop:arb_factor} shows that a strategy $\pi\in\cI_{\rm arb}\cap\strat$ necessarily satisfies $\pi^{\top}Q_{\bullet,k}=0$, for all $k=2,\ldots,\ell$. When $\ell=d$, this corresponds to a set of $d-1$ linear equations in $d$ variables. This set defines a line in $\R^d$, which we call {\em arbitrage line}. This concept will be illustrated in the two-dimensional model considered in Section \ref{sec:2dim}.
\end{rem}

In view of Theorem \ref{thm:utility}, NA$_1$ ensures the well-posedness of optimal portfolio problems. In the presence of arbitrage opportunities, the borrowing constraint \eqref{eq:theta_c} is binding for every optimal allowed strategy. This is a direct consequence of the following simple result.

\begin{lem}	\label{lem:opt_constraints}
In the context of the model of this section, suppose that ${\rm e}_1\in\cR(Q^{\top})$ and ${\rm NA}_1$ holds. Then, for every $\pi\in\strat$, there exists an element $\hat{\pi}\in\strat$ such that 
\[
\langle\hat{\pi},QY\rangle\geq\langle\pi,QY\rangle\text{ a.s.}
\qquad\text{ and }\qquad
\langle\hat{\pi},\mathbf{1}\rangle=c.
\]
Moreover, there exists a strategy $\pi^{\max}$, explicitly given by
\be	\label{eq:max_arb}
\pi^{\max} = \frac{c}{\langle(QQ^{\top})^{-1}Q_{\bullet,1},\mathbf{1}\rangle}(QQ^{\top})^{-1}Q_{\bullet,1},
\ee
such that $\langle\pi^{\max},\mathbf{1}\rangle=c$ and $\langle\pi^{\max},QY\rangle\geq\langle\pi,QY\rangle$ a.s., for all $\pi\in\cI_{\rm arb}\cap\strat$.
\end{lem}
\begin{proof}
Let $\pi$ be an arbitrary allowed strategy. Letting $\lambda:=(c-\langle\pi,\mathbf{1}\rangle)\langle(QQ^{\top})^{-1}Q_{\bullet,1},\mathbf{1}\rangle^{-1}\geq0$, define the strategy $\hat{\pi}:=\pi+\lambda(QQ^{\top})^{-1}Q_{\bullet,1}$. Clearly, it holds that $\langle\hat{\pi},\mathbf{1}\rangle=c$ and, in addition, $\langle\hat{\pi},QY\rangle=\langle\pi,QY\rangle+\lambda {\rm e}_1^{\top}Q^{\top}(QQ^{\top})^{-1}QY=\langle\pi,QY\rangle+\lambda Y_1\geq\langle\pi,QY\rangle$ a.s.
The second part of the lemma follows as a direct consequence of the  characterization \eqref{eq:arb_strategies} of the set $\cI_{\rm arb}\cap\strat$.
\end{proof}

We call {\em maximal arbitrage strategy} the strategy $\pi^{\max}$  given in \eqref{eq:max_arb}. 
Whenever NA$_1$ fails to hold (i.e.,  $\langle(QQ^{\top})^{-1}Q_{\bullet,1},\mathbf{1}\rangle\leq0$), a maximal arbitrage strategy does not exist, because arbitrage opportunities can be arbitrarily scaled.
Note that $\pi^{\max}$ is not necessarily the optimal strategy in an expected utility maximization problem of type \eqref{eq:opt_utility}. Similarly, $\pi^{\max}$ does not necessarily coincide with the num\'eraire portfolio $\rho$. This will be explicitly illustrated in Examples \ref{ex:first_example}--\ref{ex:third_example}.

\begin{rem}[On relative arbitrage]	\label{rem:rel_arb_factor}
(1)
In the context of the model of this section, let us assume that ${\rm NA}_1$ holds and ${\rm e}_1\in\cR(Q^{\top})$. Then, for $\theta\in\strat$, there exist arbitrage opportunities relative to $\theta$ if and only if $\langle\theta,\mathbf{1}\rangle<c$. Indeed, if $\langle\theta,\mathbf{1}\rangle<c$, then the existence of an arbitrage opportunity relative to $\theta$ follows from Lemma \ref{lem:opt_constraints}. Conversely, suppose that $\langle\theta,\mathbf{1}\rangle=c$ and let $\pi\in\R^d$ with $\pi-\theta\in\cI_{\rm arb}$. By Proposition \ref{prop:arb_factor}, this holds if and only if $\pi-\theta=\eta(QQ^{\top})^{-1}Q_{\bullet,1}$, for some $\eta>0$. However, since 
$
\langle\pi,\mathbf{1}\rangle
= \langle\theta,\mathbf{1}\rangle + \eta\langle(QQ^{\top})^{-1}Q_{\bullet,1},\mathbf{1}\rangle > c,
$
the strategy $\pi$ is not an allowed trading strategy.  This shows that there cannot exist arbitrage opportunities relative to $\theta$ if $\langle\theta,\mathbf{1}\rangle=c$.
In particular, there do not exist arbitrage opportunities relative to $\pi^{\max}$. 

(2)
One can also study the existence of arbitrage opportunities relative to the {\em market portfolio} $\pi^{\rm mkt}$ defined by $\pi^{\rm mkt}_i:=S^i_0/\langle S_0,\mathbf{1}\rangle$, for $i=1,\ldots,d$ (see \cite[Section 2]{FK09}). As a consequence of part (1) of this remark, arbitrage opportunities relative to the market exist if and only if $c>1$. The financial intuition is that, if $c>1$, then it is possible to invest the whole initial capital $v$ in the market portfolio, borrow an amount $v(c-1)$ from the riskless asset and invest that amount in the strategy $\pi^{\max}$, thus improving the performance of the market portfolio. The strategy $\pi^*\in\strat$ which best outperforms the market portfolio is given by $\pi^*=\pi^{\rm mkt}+\frac{c-1}{c}\pi^{\max}$. 
\end{rem}

\subsection{The case of a unit triangular matrix $Q$}	\label{sec:triang}

Let us consider the special case where $Q$ is a $(d\times d)$ upper triangular matrix with $q_{i,i}=1$, for all $i=1,\ldots,d$. In this case, the results presented in Section \ref{sec:gen_factor} can be stated explicitly in terms of the elements of $Q$. First, condition \eqref{eq:pos_prices} ensuring the positivity of asset prices can be rewritten in the following recursive form:
\be	\label{eq:pos_prices_triang}
y^{\inf}_d \geq -1
\qquad\text{ and }\qquad
y^{\inf}_i \geq -1 - \sum_{k=i+1}^d\bigl(q^+_{i,k}y^{\inf}_k-q^-_{i,k}y^{\sup}_k\bigr),
\quad\text{ for all }i=1,\ldots,d-1.
\ee
In view of \eqref{eq:theta_adm}, the set $\strat_{\rm adm}$ of admissible strategies takes the form
\be	\label{eq:theta_adm_triang}
\strat_{\rm adm} = \left\{\pi\in\R^d : \pi_1\geq0 \text{ and }
\sum_{k=2}^d\left(\Biggl(\sum_{i=1}^{k-1}\pi_iq_{i,k}+\pi_k\Biggr)^+y^{\inf}_k-\Biggl(\sum_{i=1}^{k-1}\pi_iq_{i,k}+\pi_k\Biggr)^-y^{\sup}_k\right)\geq-1\right\}.
\ee

Since ${\rm rank}(Q)=d$, the condition ${\rm e}_1\in\cR(Q^{\top})$  is automatically satisfied and, therefore, there exist arbitrage opportunities (see Proposition \ref{prop:arb_factor}).
More specifically, it holds that
\be	\label{eq:arb_strategies_triang}
\cI_{\rm arb}\cap\strat
= \bigl\{\lambda Q^{-1}_{1,\bullet} : \lambda>0 \text{ and }\lambda\langle Q^{-1}_{1,\bullet},\mathbf{1}\rangle\leq c\bigr\},
\ee
where $Q^{-1}_{1,\bullet}$ denotes the first row of the matrix $Q^{-1}$, written as a column vector. 
The following lemma gives an explicit representation of the vector $Q^{-1}_{1,\bullet}$, which determines all the arbitrage properties of the model under consideration. 

\begin{lem}	\label{lem:inverse}
In the context of the model of this section, suppose that $Q$ is a unit triangular matrix. Then,  for all $k=1,\ldots,d$, it holds that $Q^{-1}_{1,k}=\alpha_k$, where $\alpha_k$ is defined by
\[
\alpha_1 := 1
\qquad\text{and}\qquad
\alpha_k := \sum_{J\in A(k)}(-1)^{|J|-1}\prod_{l=1}^{|J|-1}q_{j_l,j_{l+1}},
\quad\text{for }k=2,\ldots,d,
\]
where $A(k)$ denotes the family of all subsets $J=\{j_1,\ldots,j_r\}\subseteq\{1,\ldots,k\}$, with $r\leq k$, such that $j_1=1$, $j_r=k$ and $j_l<j_{l+1}$, for all $l=1,\ldots,r-1$, and $|J|$ denotes the cardinality of $J$.
\end{lem}
\begin{proof}
The vector $Q^{-1}_{1,\bullet}$ is the unique solution $\pi\in\R^d$ to the linear system $Q^{\top}\pi={\rm e}_1$. Since $Q$ is a unit triangular matrix, the solution $\pi$ is characterized by $\pi_1=1$ and by the recursive relation
\be	\label{eq:recursive}
\pi_k = -\sum_{i=1}^{k-1}\pi_iq_{i,k},
\qquad\text{ for all }k=2,\ldots,d.
\ee
To prove the lemma, it suffices to show that the vector $\alpha=(\alpha_1,\ldots,\alpha_d)^{\top}$ satisfies \eqref{eq:recursive}.
To this effect, notice that, for every $k=2,\ldots,d$,
\[
-\sum_{i=1}^{k-1}\alpha_iq_{i,k}
= -q_{1,k}-\sum_{i=2}^{k-1}\sum_{J\in A(i)}(-1)^{|J|-1}\prod_{l=1}^{|J|-1}q_{j_l,j_{l+1}}q_{i,k}
= \alpha_k.
\]
This shows that $\alpha=(\alpha_1,\ldots,\alpha_d)^{\top}$ satisfies \eqref{eq:recursive} and, therefore, it holds that $Q^{-1}_{1,\bullet}=\alpha$.
\end{proof}

In view of \eqref{eq:arb_strategies_triang}, the vector $\alpha$ introduced in Lemma \ref{lem:inverse} generates all arbitrage strategies, up to a multiplicative factor depending on the borrowing constraint $c$. 
More precisely, every arbitrage strategy $\pi$ is necessarily of the form $\pi=\lambda\alpha$, with $\lambda>0$ satisfying $\lambda\langle\alpha,\mathbf{1}\rangle\leq c$, and is such that $V_1^{\pi}=1+\lambda Y_1$. Furthermore, by \eqref{eq:recursive}, all such strategies $\pi$ belong to the {\em arbitrage line} (see Remark \ref{rem:arbitrage_line}).
As an example, for $d=4$, all arbitrage strategies are proportional to
\[
\alpha 
= \begin{pmatrix}
1\\
-q_{1,2}\\
-q_{1,3}+q_{1,2}\,q_{2,3}\\
-q_{1,4}+q_{1,2}\,q_{2,4}+q_{1,3}\,q_{3,4}-q_{1,2}\,q_{2,3}\,q_{3,4}
\end{pmatrix}.
\]

In the model considered in this subsection, the condition characterizing the validity of NA$_1$ takes the simple form $\langle Q^{-1}_{1,\bullet},\mathbf{1}\rangle>0$ (see Proposition \ref{prop:arb_factor}).
As a consequence of Lemma \ref{lem:inverse}, this implies the following explicit characterization of NA$_1$:
\be	\label{eq:alpha_NA1}
\text{NA$_1$ holds }
\qquad\Longleftrightarrow\qquad
1+\sum_{J\subseteq\{1,\ldots,d\}}(-1)^{|J|-1}\prod_{l=1}^{|J|-1}q_{j_l,j_{l+1}}>0,
\ee
where the summation is taken over all sets $J=\{j_1,\ldots,j_r\}$, with $2\leq r\leq d$, such that $j_1=1$ and $j_l<j_{l+1}$, for all $l=1,\ldots,r-1$.
In view of \eqref{eq:max_arb}, the same quantity appearing on the right of \eqref{eq:alpha_NA1} represents the denominator of the maximal arbitrage strategy $\pi^{\max}$.

\subsection{A two-dimensional example with arbitrage}	\label{sec:2dim}

We now present a two-dimensional model that allows for a geometric visualization of the concepts introduced above.
Let $d=2$ and consider a pair $(Y_1,Y_2)$ of independent random variables such that $\cY_1=[0,+\infty)$ and $y^{\inf}_2<0<y^{\sup}_2$. Let
\[
Q = \begin{pmatrix}
1 & \gamma	\\
0 & 1
\end{pmatrix},
\]
with $\gamma\in\R$, and suppose that the asset returns $(R_1,R_2)$ are generated as in \eqref{eq:factors}. To ensure positive asset prices, condition \eqref{eq:pos_prices_triang} needs to be satisfied. In this example, the largest possible support of the distribution of the random factor $Y_2$ is given by
\[
\begin{cases}
y^{\inf}_2=-1 \text{ and }y^{\sup}_2=+\infty, & \text{ if }\gamma\in[0,1);\\
y^{\inf}_2=-1/\gamma \text{ and }y^{\sup}_2=+\infty, & \text{ if }\gamma\geq1;\\
y^{\inf}_2=-1 \text{ and }y^{\sup}_2=-1/\gamma, & \text{ if }\gamma<0. 
\end{cases}
\]

In view of \eqref{eq:theta_adm_triang}, a strategy $\pi=(\pi_1,\pi_2)$ is admissible if and only if
\be	\label{eq:example_adm}
\begin{cases}
\pi_1\geq0 \text{ and } -\gamma\pi_1\leq\pi_2\leq1-\gamma\pi_1, & \text{ if }\gamma\in[0,1);\\
\pi_1\geq0 \text{ and } -\gamma\pi_1\leq\pi_2\leq\gamma-\gamma\pi_1,& \text{ if }\gamma\geq1;\\
\pi_1\geq0 \text{ and } \gamma-\gamma\pi_1\leq\pi_2\leq1-\gamma\pi_1,& \text{ if }\gamma<0.
\end{cases}
\ee

In this two-dimensional setting, the borrowing constraint \eqref{eq:theta_c} takes the form $\pi_1+\pi_2\leq c$. Together with \eqref{eq:example_adm}, this constraint determines the set $\strat$ of allowed strategies.
Regardless of the values of $\gamma$ and $c$, {\em arbitrage opportunities always exist}. More specifically, it holds that
\be	\label{eq:example_arb_set}
\cI_{\rm arb}\cap\strat = \bigl\{\pi\in\R^2 : \pi_1>0,\, \pi_2=-\gamma\pi_1\text{ and }\pi_1(1-\gamma)\leq c\bigr\} \neq \emptyset.
\ee
The {\em arbitrage line} (see Remark  \ref{rem:arbitrage_line}) is described by the equation $\pi_2=-\gamma\pi_1$.
Figure \ref{fig} provides a visualization of the set $\strat$, with the arbitrage line highlighted in red.

\begin{figure}[h]	
\begin{tikzpicture}
\draw[thick,->] (-2,0) -- (6,0) node[anchor=north east] 
{$\pi_1$};
\draw[yellow, fill=yellow] (0,1) -- (3,-0.5) -- (5,-2.5) -- (0,0);
\draw[thick,->] (0,-3.5) -- (0,4) node[anchor=north east] 
{$\pi_2$};
\draw[thick] (0,1) -- (6,-2) node[anchor=south west]  
{$\pi_2=1-\gamma\pi_1$};
\draw[thick] (0,2.5) -- (6,-3.5);
\node[above right] at (0.8,1.5) {$\pi_2=c-\pi_1$};
\node[above left] at (0,1) {1};
\node[above left] at (0,2.5) {$c$};
\node[left] at (5.4,-2.8) {$\quad(\frac{c}{1-\gamma},-\frac{c\gamma}{1-\gamma})\quad$};
\draw[very thick, red] (0,0) -- (6,-3) node[anchor=south west] 
{$\textcolor{black}{\pi_2=-\gamma\pi_1}$};
\end{tikzpicture}
\caption{Geometric illustration of the set $\strat$ (yellow area), for $c=2.5$ and $\gamma=0.5$.}
\label{fig}
\end{figure}
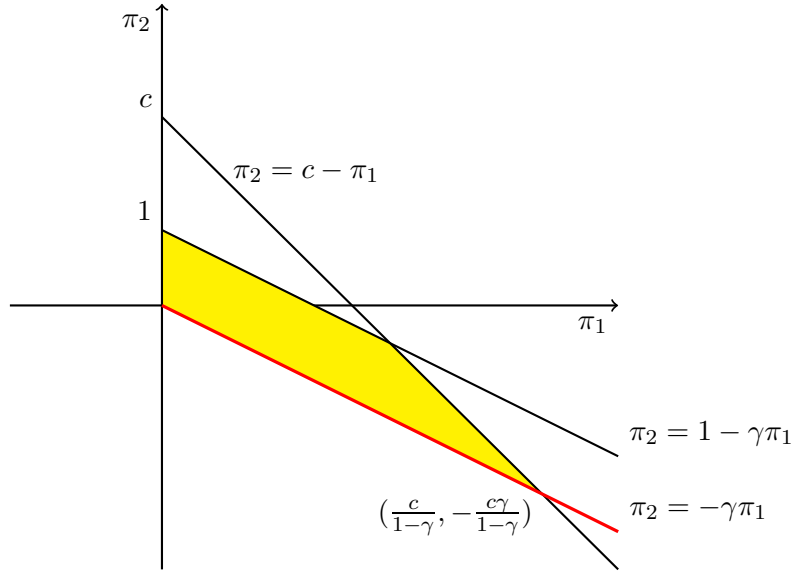

The NA$_1$ condition is satisfied if and only if $\langle Q^{-1}_{1,\bullet},\mathbf{1}\rangle>0$.  Therefore, we have that 
\[
\text{ NA$_1$ holds }
\qquad\Longleftrightarrow\qquad
\gamma<1.
\]
Indeed, from \eqref{eq:example_arb_set} we have that $\cI_{\rm arb}\cap\strathat=\emptyset$ if and only if $\gamma<1$.
Graphically, this condition corresponds to requesting that the arbitrage line intersects the borrowing constraint line (see Figure \ref{fig}), i.e., the line of equation $\pi_2=c-\pi_1$. Observe also that the set $\strat$ is compact if and only if such an intersection occurs (compare with condition {\em (iv)} in Proposition \ref{prop:NA1}).

For $\gamma<1$, all arbitrage strategies are contained in the line segment passing through the origin and the point $(\pi^{\max}_1,\pi^{\max}_2)$ characterizing the maximal arbitrage strategy and given by
\be	\label{eq:pimax_ex}
\pi^{\max}_1 = \frac{c}{1-\gamma}
\qquad\text{ and }\qquad
\pi^{\max}_2 = -\frac{c\gamma}{1-\gamma},
\ee
as follows from \eqref{eq:max_arb}.
Graphically, the strategy $\pi^{\max}$ corresponds to the point of intersection between the arbitrage line and the borrowing constraint line.
If the two lines do not intersect, then every arbitrage opportunity can be arbitrarily scaled (i.e., NA$_1$ fails to hold).

In view of Theorem \ref{thm:numeraire}, the num\'eraire portfolio $\rho$ exists if and only if $\gamma<1$. The num\'eraire portfolio may or may not coincide with the maximal arbitrage strategy $\pi^{\max}$, depending on the distributional properties of $Y_1$ and $Y_2$. For illustration, we present three simple examples.

\begin{example}	\label{ex:first_example}
Let $\gamma\in[0,1)$ and suppose that $\EE[Y_2]=0$. In this case, it holds that $\rho=\pi^{\max}$.
Indeed, let $\pi=(\pi_1,\pi_2)$ be an arbitrary strategy satisfying \eqref{eq:example_adm} and $\pi_1+\pi_2\leq c$. By Lemma \ref{lem:opt_constraints}, there exists a strategy of the form $\hat{\pi}=(\hat{\pi}_1,c-\hat{\pi}_1)$ such that $V^{\hat{\pi}}_1\geq V^{\pi}_1$ a.s. Due to \eqref{eq:example_adm}, it necessarily holds that $0\leq\hat{\pi}_1\leq c/(1-\gamma)$. Therefore, using the independence of $Y_1$ and $Y_2$ and the fact that $\EE[Y_2]=0$, we have that
\[
\EE\left[\frac{V^{\pi}_1}{V^{\pi^{\max}}_1}\right]
\leq \EE\left[\frac{V^{\hat{\pi}}_1}{V^{\pi^{\max}}_1}\right]
= \EE\left[\frac{1+\hat{\pi}_1Y_1}{1+\frac{c}{1-\gamma}Y_1}\right]
\leq 1,
\]
where the last inequality follows from the fact that $Y_1\geq0$ a.s.
This shows that the num\'eraire portfolio $\rho$ coincides with the maximal arbitrage strategy $\pi^{\max}$ given in \eqref{eq:pimax_ex}.
\end{example}

\begin{example}	\label{ex:second_example}
Let $\gamma=1/2$ and $c=1$. Suppose that $Y_1\sim{\rm Exp}(1)$ and $1+Y_2\sim{\rm Exp}(\beta)$, with $\beta>0$. In this case, for suitable values of $\beta$, the maximal arbitrage strategy is not the num\'eraire portfolio. Indeed, considering the strategy $(0,1)\in\strat$, we have that
\[
\EE\left[\frac{V^{(0,1)}_1}{V^{\pi^{\max}}_1}\right]
= \EE\left[\frac{1+Y_2}{1+2Y_1}\right]
= \frac{1}{\beta}\EE\left[\frac{1}{1+2Y_1}\right]
= \frac{\sqrt{e}}{2\beta}\int_{1/2}^{+\infty}\frac{e^{-x}}{x}\ud x 
\approx \frac{0.461}{\beta}.
\]
For any sufficiently small value of $\beta$, it holds that $\EE[V^{(0,1)}_1/V^{\pi^{\max}}_1]>1$ and, therefore,  the strategy $\pi^{\max}$ cannot be the num\'eraire portfolio in that case. Furthermore, since $V^{\pi^{\max}}_1\geq V^{\pi}_1$ a.s. for all $\pi\in\cI_{\rm arb}\cap\strat$ (see Lemma \ref{lem:opt_constraints}), the num\'eraire portfolio $\rho$ does not belong to the set of arbitrage opportunities (i.e., $\rho\notin\cI_{\rm arb}\cap\strat$).\footnote{Taking for instance $\beta=0.3$ in the example under consideration, the num\'eraire portfolio $\rho$ can be numerically computed as $\rho\approx(1.335,-0.335)\neq(2,-1)=\pi^{\max}$.}
In view of Remark \ref{rem:log_optimal}, the log-optimal strategy $\rho$ is therefore not an arbitrage strategy. Moreover, since the trading constraint \eqref{eq:theta_c} is binding for $\rho$ (as a consequence of Lemma \ref{lem:opt_constraints}), it is not allowed to improve the strategy $\rho$ by adding to it a fraction of any arbitrage strategy.

This example shows that, even in the presence of arbitrage, it is not necessarily optimal to invest in an arbitrage opportunity.
The financial intuition is that, for a logarithmic investor and sufficiently small $\beta$, the risk-reward profile of the strategy $\rho$ is more attractive than any arbitrage opportunity. Indeed, in the present example every allowed strategy $\pi=(\pi_1,\pi_2)$ satisfies
\[
V^{\pi}_1 = 1+\pi_1Y^1+(\pi_1/2+\pi_2)Y^2.
\]
Since $\pi_1/2+\pi_2\geq0$ by \eqref{eq:example_adm}, losses can only occur on the event $\{Y^2<0\}$, which happens with probability $1-\exp(-\beta)$. In view of \eqref{eq:example_arb_set}, arbitrage strategies satisfy $\pi_1/2+\pi_2=0$ and therefore eliminate the influence of the risk factor $Y^2$, with consequently no risk of losses. On the contrary, for sufficiently small $\beta$ the log-optimal strategy $\rho$ does not belong to the set $\cI_{\rm arb}$, thus implying a positive exposure to the factor $Y^2$. The financial explanation is  that the log-optimal strategy can tolerate the risk of losses in order to profit from potentially large values of $Y^2$, which are most likely for small values of $\beta$.
\end{example}

\begin{example}	\label{ex:third_example}
Let $\gamma<0$ and suppose that $\EE[Y_1]<+\infty$ and $\EE[Y_2]<+\infty$. Under these assumptions, the log-optimal portfolio $\pi^*$ exists and, therefore, it coincides with the num\'eraire portfolio $\rho$. 
Lemma \ref{lem:opt_constraints} together with \eqref{eq:example_adm} implies that $\pi^*$ is of the form $(\pi^*_1,c-\pi^*_1)$, with $\pi^*_1\in D(c,\gamma):=[\frac{(c-1)^+}{1-\gamma},\frac{c-\gamma}{1-\gamma}]$. Consider the function $g:D(c,\gamma)\rightarrow\R$ defined by
\[
g(\pi_1) := 
\EE\bigl[\log\bigl(V^{(\pi_1,c-\pi_1)}_1\bigr)\bigr] =
\EE\bigl[\log\bigl(1+\pi_1\bigl(Y_1+(\gamma-1)Y_2\bigr)+cY_2\bigr)\bigr],
\]
for $\pi_1\in D(c,\gamma)$.
Since the function $g$ is concave and $\pi^{\max}_1=c/(1-\gamma)$  belongs to the interior of the interval $D(c,\gamma)$, the log-optimal portfolio $\pi^*$ is given by $\pi^{\max}$ if and only if $g'(\pi^{\max}_1)=0$. The latter condition is equivalent to
\be	\label{eq:example_log_opt}
\EE\left[\frac{Y_1}{1+\frac{c}{1-\gamma}Y_1}\right]
= (1-\gamma) \EE\left[\frac{Y_2}{1+\frac{c}{1-\gamma}Y_1}\right].
\ee
In the present example, $\rho=\pi^{\max}$ holds if and only if condition \eqref{eq:example_log_opt} is satisfied. In particular, unlike in Example \ref{ex:first_example} where $\gamma\in[0,1)$, note that \eqref{eq:example_log_opt} cannot be satisfied if $\EE[Y_2]=0$. 
\end{example}

\section{The multi-period setting}	\label{sec:multiperiod}

In this section, we extend the analysis of Section \ref{sec:1p} to the multi-period case. We allow for convex trading constraints evolving randomly over time and prove that NA$_1$ holds in a dynamic setting if and only if it holds in each single trading period. This fundamental fact enables us to address the multi-period case by relying on arguments similar to those employed in Section \ref{sec:1p}. 
For brevity of presentation, we prove multi-period versions of only the central results characterizing market viability and NA$_1$, the remaining results and remarks admitting analogous extensions.

\subsection{Setting and trading restrictions}

Let $(\Omega,\cF,\bF,P)$ be a filtered probability space, where $\bF=(\cF_t)_{t=0,1,\ldots,T}$ and $\cF_0$ is the trivial $\sigma$-field completed by the $P$-nullsets of $\cF$, for a fixed time horizon $T\in\N$.
Similarly to Section \ref{sec:1p}, we consider $d$ risky assets and a riskless asset with constant price equal to one. The discounted prices of the $d$ risky assets are represented by the $d$-dimensional adapted process $S=(S_t)_{t=0,1,\ldots,T}$. 
For each $i=1,\ldots,d$, we assume that
\[
S^i_t = S^i_{t-1}(1+R^i_t),
\qquad\text{ for all }t=1,\ldots,T,
\]
where each random variable $R^i_t$ is $\cF_t$-measurable, satisfies $R^i_t\geq-1$ a.s. and represents the return of asset $i$ on the period $[t-1,t]$. 
For each $t=1,\ldots,T$, we denote by $\cS_t$ the $\cF_{t-1}$-conditional support of the random vector $R_t=(R^1_t,\ldots,R^d_t)^{\top}$ (i.e., the support of a regular version of the $\cF_{t-1}$-conditional distribution of $R_t$, see \cite[Definition 2.2]{BCL19}). We also denote by $\cL_t$ the smallest linear subspace of $\R^d$ containing $\cS_t$ and by $\cL_t^{\perp}$ its orthogonal complement. 
Conditional expectations are to be understood in the generalized sense (see, e.g., \cite[Section 1.4]{HWY}).

A set-valued process $A=(A_t)_{t=1,\ldots,T}$ is said to be {\em predictable} if, for each $t=1,\ldots,T$, the correspondence (set-valued mapping) $A_t$ from $\Omega$ to $\R^d$ is $\cF_{t-1}$-measurable.\footnote{We recall that a correspondence $A_t$ from $\Omega$ to $\R^d$ is $\cF_{t-1}$-measurable if, for every open subset $G\subset\R^d$, it holds that $\{\omega\in\Omega : A_t(\omega)\cap G\neq\emptyset\}\in\cF_{t-1}$, see \cite[Definition 14.1]{RW}.}
The processes $\cS=(\cS_t)_{t=1,\ldots,T}$, $\cL=(\cL_t)_{t=1,\ldots,T}$ and $\cL^{\perp}=(\cL_t^{\perp})_{t=1,\ldots,T}$ are all predictable (see \cite[Lemma 2.4]{BCL19} and \cite[Exercise 14.12-(d)]{RW}).
For each $t=1,\ldots,T$, the orthogonal projection of a vector $x\in\R^d$ on $\cL_t$ is denoted by ${\rm p}_{\cL_t}(x)$ and it is $\cF_{t-1}$-measurable (see \cite[Exercise 14.17]{RW}).

We describe trading strategies via predictable processes $\pi=(\pi_t)_{t=1,\ldots,T}$, with $\pi_t=(\pi^1_t,\ldots,\pi^d_t)^{\top}$ representing fractions of wealth held in the $d$ risky assets between time $t-1$ and time $t$. 
We denote by $V^{\pi}_t(v)$ the wealth at time $t$ generated by strategy $\pi$ starting from capital $v>0$, with
\[
V^{\pi}_0(v) = v
\qquad\text{ and }\qquad
V^{\pi}_t(v) = v\prod_{k=1}^t(1+\langle\pi_k,R_k\rangle),
\quad\text{for }t=1,\ldots,T.
\]
As in Section \ref{sec:constraints}, we define $V^{\pi}_t:=V^{\pi}_t(1)$.
A predictable strategy $\pi$ is said to be {\em admissible} if $V^{\pi}_t\geq0$ a.s., for all $t=1,\ldots,T$. Equivalently, introducing the random set
\be	\label{eq:adm_constraints}
\strat_{{\rm adm},t} 
:= \{\pi\in\R^d : \langle\pi,z\rangle\geq-1\text{ for all }z\in\cS_t\},
\qquad \text{ for }t=1,\ldots,T,
\ee
a predictable strategy $\pi$ is admissible if and only if $\pi_t\in\strat_{{\rm adm},t}$ holds a.s. for all $t=1,\ldots,T$. 
Note that, for every $(\omega,t)\in\Omega\times\{1,\ldots,T\}$, the set $\strat_{{\rm adm},t}(\omega)$ is a non-empty, closed and convex subset of $\R^d$. 
Arguing similarly as in \cite[Exercise 14.12-(e)]{RW}, it can be shown that the predictability of $\cS$ implies that the set-valued process $\strat_{\rm adm}=(\strat_{{\rm adm},t})_{t=1,\ldots,T}$ is predictable.

Trading constraints are modelled through a set-valued predictable process $\strat_{\rm c}=(\strat_{{\rm c},t})_{t=1,\ldots,T}$ such that $\strat_{{\rm c},t}(\omega)$ is a convex closed subset of $\R^d$, for all $(\omega,t)\in\Omega\times\{1,\ldots,T\}$.
Similarly as in Section \ref{sec:constraints}, we assume that $\cL^{\perp}_t(\omega)\subset\strat_{{\rm c},t}(\omega)$, for all $(\omega,t)\in\Omega\times\{1,\ldots,T\}$.
The family of {\em allowed strategies} is given by all $\R^d$-valued predictable processes $\pi=(\pi_t)_{t=1,\ldots,T}$ such that $\pi_t$ belongs a.s. to $\strat_t:=\strat_{{\rm adm},t}\cap\strat_{{\rm c},t}$, for all $t=1,\ldots,T$. 
Note that, as a consequence of \cite[Proposition 14.11]{RW}, the set-valued process $\strat=(\strat_t)_{t=1,\ldots,T}$ is predictable.
For brevity of notation, we shall simply write $\pi\in\strat$ to denote that a trading strategy $\pi$ is allowed.
For each $(\omega,t)\in\Omega\times\{1,\ldots,T\}$, the set $\strathat_t(\omega)$ is defined as the recession cone of $\strat_t(\omega)$. The set-valued process $\strathat=(\strathat_t)_{t=1,\ldots,T}$ is predictable, as a consequence of the predictability of $\strat$ together with \cite[Exercise 14.21]{RW}, and admits the same financial interpretation as the recession cone $\strathat$ introduced in a single-period setting  in Section \ref{sec:constraints}.

\begin{rem}	\label{rem:random_constraints}
Trading constraints evolving {\em randomly} over time arise naturally as a consequence of the admissibility requirement and are not  purely motivated by mathematical generality, as pointed out also in \cite{KaratzasKardaras07}. 
Indeed, admissibility requires that, for each $t=1\ldots,T$, the strategy $\pi_t$ is chosen at time $t-1$ in such a way that $\langle\pi_t,R_t\rangle\geq-1$ a.s., conditionally on the information available up to time $t-1$. 
Therefore, as becomes apparent from \eqref{eq:adm_constraints}, the randomness of $\strat_{{\rm adm},t}$ is due to the fact that the $\cF_{t-1}$-conditional support $\cS_t$ of $R_t$ and, therefore, the set of admissible strategies $\pi_t$ may depend on the realizations of the asset returns $(R_1,\ldots,R_{t-1})$.
Consequently, even in the presence of deterministic trading constraints $\strat_{\rm c}$, the set-valued process $\strat$ of allowed strategies is a deterministic process only in the special case where the asset returns $(R_t)_{t=1,\ldots,T}$ form a sequence of serially independent random vectors.
\end{rem}

\subsection{Arbitrage concepts}

An allowed strategy $\pi\in\strat$ is said to be an {\em arbitrage opportunity} if 
\be	\label{eq:arb_multi_per}
P(V^{\pi}_T\geq 1)=1
\qquad\text{ and }\qquad
P(V^{\pi}_T>1)>0.
\ee
We say that {\em no classical arbitrage} holds if there does not exist a strategy $\pi\in\strat$ satisfying \eqref{eq:arb_multi_per}. 
For $t=1,\ldots,T$, we denote by $L^0_+(\cF_t)$  the family of non-negative $\cF_t$-measurable random variables.
Definition \ref{def:NA1} can be naturally extended to a multi-period setting as follows.

\begin{defn}	\label{def:NA1_mp}
A random variable $\xi\in L^0_+(\cF_T)$ with $P(\xi>0)>0$ is said to be an {\em arbitrage of the first kind} if $v(\xi)=0$, where
$
v(\xi) := \inf\{v>0 : \exists\;\pi\in\strat \text{ such that }V^{\pi}_T(v)\geq\xi\text{ a.s.}\}
$.\\
{\em No arbitrage of the first kind} (NA$_1$) holds if, for every $\xi\in L^0_+(\cF_T)$, $v(\xi)=0$ implies $\xi=0$ a.s.
\end{defn}

As a preliminary to the statement of the next proposition, we define, for each $t=1,\ldots,T$,
\[
\cI_{{\rm arb},t} := \{\pi\in\R^d : \langle\pi,z\rangle\geq0\text{ for all }z\in\cS_t\}\setminus\cL_t^{\perp}.
\]
By \cite[Exercise 14.12-(e)]{RW}, the random set $\cI_{{\rm arb},t}$ is $\cF_{t-1}$-measurable, for all $t=1,\ldots,T$. 
For a random variable $\zeta\in L^0_+(\cF_t)$, we define its super-hedging value at time $t-1$ by
\[
{\rm v}_{t-1}(\zeta) := \essinf\bigl\{x\in L^0_+(\cF_{t-1}) : \exists\;h\in L^0(\cF_{t-1};\strat_t) \text{ such that }x(1+\langle h,R_t\rangle)\geq\zeta\text{ a.s.}\bigr\},
\]
where $L^0(\cF_{t-1};\strat_t)$ denotes the family of $\cF_{t-1}$-measurable random vectors $h:\Omega\rightarrow\R^d$ 
such that $P(h\in\strat_t)=1$.

For the usual concept of no classical arbitrage, it is well-known that absence of arbitrage  in a multi-period setting is equivalent to  absence of arbitrage opportunities in each single trading period (see, e.g., \cite[Proposition 5.11]{FollmerSchied}). In the next proposition, we prove that an analogous property holds for NA$_1$ and we also provide several equivalent characterizations.

\begin{prop}	\label{prop:NA1_mp}
The following are equivalent:
\begin{enumerate}[(i)]
\item the {\rm NA$_1$} condition holds;
\item there does not exist a strategy $\pi\in\strathat$ satisfying \eqref{eq:arb_multi_per};
\item for every $t=1,\ldots,T$ and $\zeta\in L^0_+(\cF_t)$, ${\rm v}_{t-1}(\zeta)=0$ a.s. implies $\zeta=0$ a.s.;
\item $\cI_{{\rm arb},t}\cap\strathat_t=\emptyset$ a.s., for all $t=1,\ldots,T$;
\item $\strathat_t=\cL^{\perp}_t$ a.s., for all $t=1,\ldots,T$;
\item the set $\strat_t\cap\cL_t$ is a.s. bounded (and, hence, compact), for all $t=1,\ldots,T$.
\end{enumerate}
\end{prop}
\begin{proof}
$(i)\Rightarrow(iii)$: by way of contradiction, assume that NA$_1$ holds and suppose that, for some $t=1,\ldots,T$, there exists $\zeta\in L^0_+(\cF_t)$ such that ${\rm v}_{t-1}(\zeta)=0$ a.s. and $P(\zeta>0)>0$. In this case, for every $v>0$, one can find  $h\in L^0(\cF_{t-1};\strat_t)$ such that $v(1+\langle h,R_t\rangle)\geq\zeta$ a.s.
Define then the strategy $\pi=(\pi_s)_{s=1,\ldots,T}$ by $\pi_s:=h$ if $s=t$ and $\pi_s:=0$ otherwise. With this definition, it holds that $\pi\in\strat$ and  $V^{\pi}_T(v)=v(1+\langle h,R_t\rangle)\geq\zeta$ a.s., contradicting the validity of NA$_1$.\\
$(iii)\Rightarrow(iv)$: we adapt to the present setting the arguments of \cite[Section 5]{KaratzasKardaras07}.
By way of contradiction, assume that $(iii)$ holds and let $P(\cI_{{\rm arb},t}\cap\strathat_t\neq\emptyset)>0$, for some $t=1,\ldots,T$. For each $n\in\N$, define the $\cF_{t-1}$-measurable random set
\[
\cI^n_{{\rm arb},t} :=
\left\{\pi\in\R^d : \langle\pi,z\rangle\geq0 \text{ for all $z\in\cS_t$ and }\EE\left[\frac{\langle\pi,R_t\rangle}{1+\langle\pi,R_t\rangle}\bigg|\cF_{t-1}\right]\geq 1/n\right\} \subset \cI_{{\rm arb},t}.
\]
We have that $\cI_{{\rm arb},t}\cap\strathat_t\neq\emptyset$ if and only if $\cI^n_{{\rm arb},t}\cap\strathat_t\neq\emptyset$ for all large enough $n\in\N$ (see \cite[Lemma 5.1]{KaratzasKardaras07}). Hence, there exists a sufficiently large $n\in\N$ such that $P(\cI^n_{{\rm arb},t}\cap\strathat_t\neq\emptyset)>0$.
It can be easily checked that the set $\cI^n_{{\rm arb},t}(\omega)\cap\strathat_t(\omega)$ is closed and convex, for all $\omega\in\Omega$. Therefore, by \cite[Corollary 14.6]{RW}, there exists an $\cF_{t-1}$-measurable random vector $\pi^n_t:\Omega\rightarrow\R^d$ such that $\pi^n_t(\omega)\in\cI^n_{{\rm arb},t}(\omega)\cap\strathat_t(\omega)$ when $\cI^n_{{\rm arb},t}(\omega)\cap\strathat_t(\omega)\neq\emptyset$ and $\pi^n_t(\omega)=0$ when $\cI^n_{{\rm arb},t}(\omega)\cap\strathat_t(\omega)=\emptyset$. 
The random variable $\zeta:=\langle\pi^n_t,R_t\rangle$ belongs to $L^0_+(\cF_t)$ and satisfies $P(\zeta>0)>0$. Moreover, since $\pi^n_t\in\strathat_t$ a.s., it holds that $\pi^n_t/v\in\strat_t$ a.s., for all $v>0$. Noting that $v(1+\langle\pi^n_t/v,R_t\rangle)>\zeta$ a.s., this implies that ${\rm v}_{t-1}(\zeta)=0$ a.s., thus contradicting property $(iii)$.\\
$(ii)\Leftrightarrow(iv)$: this equivalence follows by the same arguments used in \cite[Proposition 5.11]{FollmerSchied}, together with the construction of $\pi^n_t$ performed in the previous step of the proof.\\
$(iv)\Rightarrow(v)\Rightarrow(vi)$: these implications can be proved  as in Proposition \ref{prop:NA1}. \\
$(vi)\Rightarrow(i)$: by way of contradiction, let $\xi\in L^0_+(\cF_T)$ with $P(\xi>0)>0$ and suppose that, for all $n\in\N$, there exists an allowed strategy $\pi^n\in\strat$ such that $V^{\pi^n}_T(1/n)\geq\xi$ a.s. Then, it holds that $1+\prod_{t=1}^T\langle{\rm p}_{\cL_t}(\pi^n_t),R_t\rangle\geq n\xi$ a.s., for all $n\in\N$. 
Similarly as in the proof of Proposition \ref{prop:NA1}, the fact that $P(\xi>0)>0$ contradicts the a.s. boundedness of the sets $\strat_t\cap\cL_t$, for $t=1,\ldots,T$.
\end{proof}

Proposition \ref{prop:NA1_mp} shows that, in a multi-period setting, NA$_1$ is equivalent to the absence of arbitrarily scalable arbitrage opportunities (property $(ii)$) as well as to the absence of arbitrage of the first kind in each single trading period (property $(iii)$). Properties $(iv)$--$(vi)$ can be interpreted similarly to the analogous properties discussed in Section \ref{sec:arbitrage_concepts}.
Note also that NA$_1$ is equivalent to no classical arbitrage if the constraint process $\strat_{{\rm c}}$ is cone-valued (see Remark \ref{rem:compare_arb}).

\begin{rem}	\label{rem:bound_meas}
Property $(vi)$ in Proposition \ref{prop:NA1_mp} implies that, for each $t=1,\ldots,T$, there exists an $\cF_{t-1}$-measurable random variable $H_t$ such that $\|\pi\|\leq H_t$ a.s., for all $\pi\in L^0(\cF_{t-1};\strat_t\cap\cL_t)$.
The $\cF_{t-1}$-measurability of $H_t$ follows from the closedness and $\cF_{t-1}$-measurability of $\strat_t\cap\cL_t$.
\end{rem}

\subsection{Market viability and fundamental theorems}

We proceed to characterize NA$_1$ in terms of the solvability of portfolio optimization problems, extending Theorem \ref{thm:utility} to the multi-period setting. 
In view of Proposition \ref{prop:NA1_mp}, the NA$_1$ condition admits a local description. 
By employing a dynamic programming approach, this allows reducing a portfolio optimization problem to a sequence of one-period problems, to which we can apply techniques analogous to those used in the proof of Theorem \ref{thm:utility}.
This approach is inspired by \cite{RasonyStettner06}, where the implication $(i)\Rightarrow(ii)$ of the following theorem has been proved under no classical arbitrage for an unconstrained market.
In comparison to \cite{RasonyStettner06}, we allow for convex trading constraints and base our analysis on the minimal NA$_1$ condition.
Similarly as in Section \ref{sec:ftaps}, we denote by $\cU$ the set of all functions $U:\Omega\times\R_+\rightarrow\R\cup\{-\infty\}$ such that $U(\cdot,x)$ is $\cF_T$-measurable and bounded from below, for every $x>0$, and $U(\omega,\cdot)$ is continuous, strictly increasing and concave, for a.e. $\omega\in\Omega$. 

\begin{thm}	\label{thm:utility_mp}
The following are equivalent:
\begin{enumerate}[(i)]
\item the {\em NA}$_1$ condition holds;
\item for every $U\in\cU$ such that $\sup_{\pi\in\strat}\EE[U^+(V^{\pi}_T)]<+\infty$, there exists an allowed strategy $\pi^*\in\strat\cap\cL$ such that
\[
\EE\bigl[U(V^{\pi^*}_T)\bigr] 
= \underset{\pi\in\strat}{\sup}\,\EE\bigl[U(V^{\pi}_T)\bigr].
\]
\end{enumerate}
\end{thm}
\begin{proof}
$(i)\Rightarrow(ii)$: 
suppose that NA$_1$ holds and let $U\in\cU$ be such that $\sup_{\pi\in\strat}\EE[U^+(V^{\pi}_T)]<+\infty$. 
Since $U\in\cU$, it holds that $\sup_{\pi\in\strat}\EE[U^+(xV^{\pi}_T)]<+\infty$ for all $x\geq0$.
The existence of an optimal strategy $\pi^*\in\strat\cap\cL$ will be shown in a constructive way by applying dynamic programming.
For all $(\omega,x)\in\Omega\times\R_+$, define $U_T(\omega,x):=U(\omega,x)$ and, for $t=0,1,\ldots,T-1$,
\be	\label{eq:dyn_prog}
U_t(\omega,x) := \underset{\pi_{t+1}\in L^0(\cF_t;\strat_{t+1}\cap\cL_{t+1})}{\esssup}\EE\left[U_{t+1}\bigl(\omega,x(1+\langle\pi_{t+1},R_{t+1}(\omega)\rangle)\bigr)\big|\cF_t\right](\omega),
\ee
taking a regular version of the conditional expectation (the existence of the conditional expectation will follow from the proof below).\footnote{In the following, for simplicity of notation, we shall omit to denote explicitly the dependence on $\omega$ in $U_t(\omega,x)$.}
Proceeding by backward induction, let $t<T$ and suppose that $U_{t+1}\in\cU$ and
\be	\label{eq:ass_mp_finite}
\sup_{\pi_{t+1}\in L^0(\cF_t;\strat_{t+1}\cap\cL_{t+1})}\EE\left[U^+_{t+1}\bigl(x(1+\langle\pi_{t+1},R_{t+1}\rangle)\bigr)\right] < +\infty,
\qquad\text{ for all }x\geq0.
\ee
These hypotheses are satisfied by assumption for $t=T-1$ and will be shown inductively for all $t<T-1$.
Since the family $\{\EE[U_{t+1}(x(1+\langle\pi_{t+1},R_{t+1}\rangle))|\cF_t];\pi_{t+1}\in L^0(\cF_t;\strat_{t+1}\cap\cL_{t+1})\}$ is directed upward, for all $x>0$ there exists a sequence $(\pi_{t+1}^n(x))_{n\in\N}$ with values in $\strat_{t+1}\cap\cL_{t+1}$ such that
\be	\label{eq:sequence_esssup}
\lim_{n\rightarrow+\infty}\EE\left[U_{t+1}\bigl(x(1+\langle\pi^n_{t+1}(x),R_{t+1}\rangle)\bigr)\big|\cF_t\right] = U_t(x)
\text{ a.s.} 
\ee
As a consequence of NA$_1$, the set $\strat_{t+1}\cap\cL_{t+1}$ is closed and a.s. bounded (see Proposition \ref{prop:NA1_mp}). Therefore, by \cite[Lemma 1.64]{FollmerSchied}, there exists a subsequence $(\pi^{n_k}_{t+1}(x))_{k\in\N}$ converging a.s. to an element $\hat{\pi}_{t+1}(x)\in L^0(\cF_t;\strat_{t+1}\cap\cL_{t+1})$.
By the same arguments used in the proof of the implication $(i)\Rightarrow(ii)$ of Theorem \ref{thm:utility} (but carried out conditionally on $\cF_t$, see also \cite[Lemma 2.3]{RasonyStettner06}), the boundedness of $\strat_{t+1}\cap\cL_{t+1}$ (see Remark \ref{rem:bound_meas}), the properties of $U_{t+1}$ and \eqref{eq:ass_mp_finite} together imply the existence of an $\cF_{t+1}$-measurable integrable random variable $\zeta_{t+1}$ such that
\be	\label{eq:upperbound_mp}
U^+_{t+1}\bigl(x(1+\langle\pi_{t+1},R_{t+1}\rangle)\bigr) \leq \zeta_{t+1},
\qquad\text{ for all }\pi_{t+1}\in L^0(\cF_t;\strat_{t+1}\cap\cL_{t+1}).
\ee
Therefore, an application of Fatou's lemma, together with the continuity of $U_{t+1}$, yields that
\begin{align*}
\limsup_{k\rightarrow+\infty}\EE\bigl[U_{t+1}\bigl(x(1+\langle\pi^{n_k}_{t+1}(x),R_{t+1}\rangle)\bigr)\big|\cF_t\bigr]
&\leq \EE\Bigl[\limsup_{k\rightarrow+\infty}U_{t+1}\bigl(x(1+\langle\pi^{n_k}_{t+1}(x),R_{t+1}\rangle)\bigr)\Big|\cF_t\Bigr]	\\
&= \EE\left[U_{t+1}\bigl(x(1+\langle\hat{\pi}_{t+1}(x),R_{t+1}\rangle)\bigr)\big|\cF_t\right].
\end{align*}
Together with \eqref{eq:sequence_esssup}, this shows that 
\be	\label{eq:optimizer_mp}
U_t(x)=\EE\bigl[U_{t+1}\bigl(x(1+\langle\hat{\pi}_{t+1}(x),R_{t+1}\rangle)\bigr)\big|\cF_t\bigr].
\ee
Condition \eqref{eq:ass_mp_finite} implies that $U_t(x)<+\infty$ a.s., for all $x\geq0$, thus proving the well-posedness of \eqref{eq:dyn_prog}.
Moreover, the same arguments employed in \cite[Lemma 2.5]{RasonyStettner06} allow to show that the optimizer $\hat{\pi}_{t+1}(x)$ can be chosen $\cF_t\otimes\cB(\R_+)$-measurable.\footnote{While \cite{RasonyStettner06} work under no classical  arbitrage and do not consider trading constraints, an inspection of the proof of their Lemma 2.5 shows that only the a.s. boundedness of the set of allowed strategies is needed. In our context, the latter property holds under NA$_1$ as a consequence of Proposition \ref{prop:NA1_mp}.}
Since the set $\strat_{t+1}\cap\cL_{t+1}$ is convex and we assumed that $U_{t+1}\in\cU$, the function $U_t(\omega,\cdot)$ inherits the strict increasingness and concavity of $U_{t+1}(\omega,\cdot)$, for a.e. $\omega\in\Omega$.
Furthermore, $U_t(x)\geq\EE[U_{t+1}(x)|\cF_t]$ and, therefore, $U_t(x)$ is a.s. bounded from below, for every $x>0$.
In particular, this implies that $U_t(x)$ is a.s. finite valued for all $x>0$ and, by concavity, continuous on $(0,+\infty)$. 
To prove continuity at $x=0$, note that $U_t(0)\leq\liminf_{n\rightarrow+\infty}U_t(1/n)$. On the other hand, using \eqref{eq:optimizer_mp}, it holds that
\begin{align*}
\limsup_{n\rightarrow+\infty}U_t(1/n)
&= \limsup_{n\rightarrow+\infty}\EE\bigl[U_{t+1}\bigl((1/n)(1+\langle\hat{\pi}_{t+1}(1/n),R_{t+1}\rangle)\bigr)\big|\cF_t\bigr]	\\
&\leq \EE\Bigl[\limsup_{n\rightarrow+\infty}U_{t+1}\bigl((1/n)(1+\langle\hat{\pi}_{t+1}(1/n),R_{t+1}\rangle)\bigr)\Big|\cF_t\Bigr]
= \EE[U_{t+1}(0)|\cF_t] = U_t(0),
\end{align*}
where, similarly as above, the inequality follows from Fatou's lemma using \eqref{eq:upperbound_mp} and the second equality follows from the continuity of $U_{t+1}$ together with the a.s. boundedness of $\strat_{t+1}\cap\cL_{t+1}$.
We have thus shown that $U_t\in\cU$. 
To complete the proof of the inductive hypothesis, it remains to show that \eqref{eq:ass_mp_finite} holds true for each $t<T-1$.
For every $x>0$ and $\pi_t\in L^0(\cF_{t-1};\strat_t\cap\cL_t)$, using repeatedly \eqref{eq:optimizer_mp} and iterated conditioning, we have that
\be	\label{eq:proof_ass_mp_finite}
\EE\bigl[U_t^+\bigl(x(1+\langle\pi_t,R_t\rangle)\bigr)\bigr]
\leq \EE\biggl[U^+\biggl(x(1+\langle\pi_t,R_t\rangle)\prod_{k=1}^{T-t}(1+\langle\hat{\pi}_{t+k}(V_{t+k-1}),R_{t+k}\rangle)\biggr)\biggr],
\ee
with $V_t:=x(1+\langle\pi_t,R_t\rangle)$ and $V_{t+k}:=V_{t+k-1}(1+\langle\hat{\pi}_{t+k}(V_{t+k-1}),R_{t+k}\rangle)$, for $k=1,\ldots,T-t$.
Since $\sup_{\pi\in\strat}\EE[U^+(xV^{\pi}_T)]<+\infty$, inequality  \eqref{eq:proof_ass_mp_finite} implies the validity of \eqref{eq:ass_mp_finite}, for all $t=0,1,\ldots,T-2$.
Finally, the optimal strategy $\pi^*=(\pi^*_t)_{t=1,\ldots,T}\in\strat\cap\cL$ is defined recursively by 
\[
\pi^*_t:=\hat{\pi}_t(V^{\pi^*}_{t-1}),
\qquad\text{where }\;
V^{\pi^*}_t=V^{\pi^*}_{t-1}(1+\langle\pi^*_t,R_t\rangle),
\text{ for all }t=1,\ldots,T, 
\,\text{ and }\,V^{\pi^*}_0=1.
\]
The optimality of $\pi^*$ follows by noting that, for any strategy $\pi\in\strat$, 
\[
\EE[U(V^{\pi}_T)]
\leq \EE[U_{T-1}(V^{\pi}_{T-1})]
\leq \ldots \leq U_0(1)
= \EE[U_1(V^{\pi^*}_1)]
= \ldots = \EE[U(V^{\pi^*}_T)].
\]
$(ii)\Rightarrow(i)$: 
in view of Proposition \ref{prop:NA1_mp}, this implication follows by the same argument used for proving the implication $(ii)\Rightarrow(i)$ in Theorem \ref{thm:utility}.
\end{proof}

To the best of our knowledge, Theorem \ref{thm:utility_mp} provides the most general characterization of market viability for discrete-time models under random convex constraints.

In the following definition, for $\pi\in\strat$, we denote by $V^{\pi}$ the stochastic process $(V^{\pi}_t)_{t=0,1,\ldots,T}$.

\begin{defn}	\label{def:deflator_mp}
An adapted stochastic process $Z=(Z_t)_{t=0,1,\ldots,T}$ satisfying $Z_t>0$ a.s. for all $t=1,\ldots,T$ and $Z_0=1$ is said to be a {\em supermartingale deflator} if $ZV^{\pi}$ is a supermartingale, for all $\pi\in\strat$. 
The set of all supermartingale deflators is denoted by $\cD$. 
An allowed strategy $\rho\in\strat$ is said to be a {\em num\'eraire portfolio} if $1/V^{\rho}\in\cD$, i.e., if $V^{\pi}/V^{\rho}$ is a supermartingale.
\end{defn}

We now prove a version of the fundamental theorem of asset pricing based on NA$_1$ in the presence of convex constraints, extending Theorem \ref{thm:numeraire} to the multi-period case. In a continuous-time semimartingale setting, the general version of this result is given in \cite[Theorem 4.12]{KaratzasKardaras07}.  
By relying on the same approach adopted in the proof of Theorem \ref{thm:numeraire}, we can give a simple and short proof in a general discrete-time setting.

\begin{thm}	\label{thm:numeraire_mp}
The following are equivalent:
\begin{enumerate}[(i)]
\item the {\em NA}$_1$ condition holds;
\item $\cD\neq\emptyset$;
\item there exists the num\'eraire portfolio.
\end{enumerate}
\end{thm}
\begin{proof}
$(i)\Rightarrow(iii)$: 
let $t\in\{1,\ldots,T\}$ and consider a family $(f_n)_{n\in\N}$ of measurable functions such that $f_n:\R^d\rightarrow(0,1]$ and $\EE[\log(1+\|R_t\|)f_n(R_t)]<+\infty$, for each $n\in\N$, and $f_n\nearrow1$ as $n\rightarrow+\infty$ (see the proof of Theorem \ref{thm:numeraire}). 
For each $n\in\N$, let $U_{t,n}(\omega,x):=\log(x)f_n(R_t(\omega))$, for all $(\omega,x)\in\Omega\times(0,+\infty)$.
For each $n\in\N$, it holds that $U_{t,n}\in\cU$.
By Proposition \ref{prop:NA1_mp}, NA$_1$ implies that $\strat_t\cap\cL_t$ is a.s. bounded and, therefore, inequality \eqref{eq:ineq_log} conditionally on $\cF_{t-1}$ implies that
$\esssup_{\pi_t\in L^0(\cF_{t-1};\strat_t\cap\cL_t)}\EE[U^+_{t,n}(1+\langle\pi_t,R_t\rangle)|\cF_{t-1}]<+\infty$ a.s.
Using again the boundedness of $\strat_t\cap\cL_t$, this can be shown to imply the existence of an element $\rho^n_t\in L^0(\cF_{t-1};\strat_t\cap\cL_t)$ such that
\[
\EE\bigl[U_{t,n}(1+\langle\rho_t^n,R_t\rangle)\big|\cF_{t-1}\bigr]
= \underset{\pi_t\in L^0(\cF_{t-1};\strat_t\cap\cL_t)}{\esssup}\EE\bigl[U_{t,n}(1+\langle\pi_t,R_t\rangle)\big|\cF_{t-1}\bigr]
\text{ a.s.}
\]
By the same reasoning as in \eqref{eq:rel_log_opt}-\eqref{eq:rel_log_opt_2} (now conditionally on $\cF_{t-1}$), we obtain that
\[
\EE\left[\frac{\langle\pi_t-\rho_t^n,R_t\rangle}{1+\langle\rho_t^n,R_t\rangle}f_n(R_t)\bigg|\cF_{t-1}\right] \leq 0
\text{ a.s.},
\qquad\text{ for all }\pi_t\in\strat_t\text{ and }n\in\N.
\]
Since $\strat_t\cap\cL_t$ is bounded and closed, we can assume that $(\rho_t^n)_{n\in\N}$ converges a.s. to an element $\rho_t\in L^0(\cF_{t-1};\strat_t\cap\cL_t)$ as $n\rightarrow+\infty$ (up to passing to a suitable subsequence, see \cite[Lemma 1.64]{FollmerSchied}). Since $f_n\nearrow1$ as $n\rightarrow+\infty$, an application of Fatou's lemma gives that
\[
\EE\left[\frac{\langle\pi_t-\rho_t,R_t\rangle}{1+\langle\rho_t,R_t\rangle}\bigg|\cF_{t-1}\right] \leq 0
\text{ a.s.},
\qquad\text{ for all }\pi_t\in L^0(\cF_{t-1};\strat_t).
\]
Let $\pi=(\pi_t)_{t=1,\ldots,T}\in\strat$. Then, for each $t\in\{1,\ldots,T-1\}$, the last inequality implies that
\[
\EE\left[\frac{V^{\pi}_t}{V^{\rho}_t}\bigg|\cF_{t-1}\right]
= \frac{V^{\pi}_{t-1}}{V^{\rho}_{t-1}}\,\EE\left[\frac{1+\langle\pi_t,R_t\rangle}{1+\langle\rho_t,R_t\rangle}\bigg|\cF_{t-1}\right]
\leq \frac{V^{\pi}_{t-1}}{V^{\rho}_{t-1}}
\text{ a.s.},
\]
thus proving that the strategy $\rho=(\rho_t)_{t=1,\ldots,T}$ corresponds to the num\'eraire portfolio.
\\
$(iii)\Rightarrow(ii)$: 
this implication is immediate by Definition \ref{def:deflator_mp}.\\
$(ii)\Rightarrow(i)$: 
this implication follows by the same argument used in the proof of  Theorem \ref{thm:numeraire}.
\end{proof}

Finally, we mention that the proof of Theorem \ref{thm:arb} can be similarly extended to the multi-period case, thus providing a utility maximization proof of the fundamental theorem of asset pricing for no classical arbitrage, in the spirit of \cite{Rogers94} (see also \cite[Section 2.1.4]{KabSaf}). 
Theorem \ref{thm:duality} also admits a direct extension to the multi-period setting, with an identical statement.

\bibliographystyle{alpha-abbrvsort}
\bibliography{bib_FR}

\end{document}